\documentclass[aip,rsi,preprint,groupedaddress,superscriptaddress,floatfix,showpacs]{revtex4}

\usepackage{hyperref,color}
\usepackage{amsmath,amsfonts,amsthm,graphicx,harpoon,amssymb,mathtools,multirow,diagbox,tikz,epsfig,subfigure}
\usepackage[center]{titlesec}
\titleformat{\section}[hang]{\Large\bfseries\filcenter}{}{1em}{}
\titleformat{\subsection}[hang]{\bfseries}{}{1em}{}

\newtheorem{conjecture}{Conjecture}

\newtheorem{definition}{Definition}
\newtheorem{remark}{Remark}
\newtheorem{lemma}{Lemma}
\newtheorem{theorem}{Theorem}
\newtheorem{proposition}{Proposition}
\newtheorem{corollary}{Corollary}

\newcommand{\bra}[1]{\langle{#1}|}
\newcommand{\ket}[1]{|{#1}\rangle}
\newcommand{\proj}[1]{|{#1}\rangle \langle {#1}|}
\newcommand{\ketbra}[2]{|{#1}\rangle \! \langle{#2}|}

\newcommand{\abs}[1]{\left\lvert {#1} \right\rvert}
\makeatletter

\newcommand{\Rmnum}[1]{\expandafter\@slowromancap\romannumeral  #1@}
\makeatother

\def\bma{\begin{bmatrix}}
\def\ema{\end{bmatrix}}
\def\rank{\mathop{\rm rank}}
\def\min{\mathop{\rm min}}
\def\diag{\mathop{\rm diag}}
\def\tr{{\rm Tr}}

\def\dg{\dagger}

\def\ox{\otimes}

\def\a{\alpha}
\def\b{\beta}
\def\g{\gamma}
\def\d{\delta}

\def\t{\theta}

\def\l{\lambda}
\def\m{\mu}
\def\n{\nu}
\def\x{\xi}
\def\p{\pi}
\def\r{\rho}
\def\s{\sigma}

\def\ph{\varphi}

\def\ps{\psi}

\newcommand{\nc}{\newcommand}

 \nc{\bbA}{\mathbb{A}} \nc{\bbB}{\mathbb{B}} \nc{\bbC}{\mathbb{C}}
 \nc{\bbD}{\mathbb{D}} \nc{\bbE}{\mathbb{E}} \nc{\bbF}{\mathbb{F}}
 \nc{\bbG}{\mathbb{G}} \nc{\bbH}{\mathbb{H}} \nc{\bbI}{\mathbb{I}}
 \nc{\bbJ}{\mathbb{J}} \nc{\bbK}{\mathbb{K}} \nc{\bbL}{\mathbb{L}}
 \nc{\bbM}{\mathbb{M}} \nc{\bbN}{\mathbb{N}} \nc{\bbO}{\mathbb{O}}
 \nc{\bbP}{\mathbb{P}} \nc{\bbQ}{\mathbb{Q}} \nc{\bbR}{\mathbb{R}}
 \nc{\bbS}{\mathbb{S}} \nc{\bbT}{\mathbb{T}} \nc{\bbU}{\mathbb{U}}
 \nc{\bbV}{\mathbb{V}} \nc{\bbW}{\mathbb{W}} \nc{\bbX}{\mathbb{X}}
 \nc{\bbZ}{\mathbb{Z}}

 \nc{\bA}{{\bf A}} \nc{\bB}{{\bf B}} \nc{\bC}{{\bf C}}
 \nc{\bD}{{\bf D}} \nc{\bE}{{\bf E}} \nc{\bF}{{\bf F}}
 \nc{\bG}{{\bf G}} \nc{\bH}{{\bf H}} \nc{\bI}{{\bf I}}
 \nc{\bJ}{{\bf J}} \nc{\bK}{{\bf K}} \nc{\bL}{{\bf L}}
 \nc{\bM}{{\bf M}} \nc{\bN}{{\bf N}} \nc{\bO}{{\bf O}}
 \nc{\bP}{{\bf P}} \nc{\bQ}{{\bf Q}} \nc{\bR}{{\bf R}}
 \nc{\bS}{{\bf S}} \nc{\bT}{{\bf T}} \nc{\bU}{{\bf U}}
 \nc{\bV}{{\bf V}} \nc{\bW}{{\bf W}} \nc{\bX}{{\bf X}}
 \nc{\bZ}{{\bf Z}}

\nc{\cA}{{\cal A}} \nc{\cB}{{\cal B}} \nc{\cC}{{\cal C}}
\nc{\cD}{{\cal D}} \nc{\cE}{{\cal E}} \nc{\cF}{{\cal F}}
\nc{\cG}{{\cal G}} \nc{\cH}{{\cal H}} \nc{\cI}{{\cal I}}
\nc{\cJ}{{\cal J}} \nc{\cK}{{\cal K}} \nc{\cL}{{\cal L}}
\nc{\cM}{{\cal M}} \nc{\cN}{{\cal N}} \nc{\cO}{{\cal O}}
\nc{\cP}{{\cal P}} \nc{\cQ}{{\cal Q}} \nc{\cR}{{\cal R}}
\nc{\cS}{{\cal S}} \nc{\cT}{{\cal T}} \nc{\cU}{{\cal U}}
\nc{\cV}{{\cal V}} \nc{\cW}{{\cal W}} \nc{\cX}{{\cal X}}
\nc{\cZ}{{\cal Z}}

\nc{\cnn}{{\cal NN}}


\begin{document}

\title{Entangling power of two-qubit unitary operations}

\author{Yi Shen}\email[]{yishen@buaa.edu.cn}
\affiliation{School of Mathematics and Systems Science, Beihang University, Beijing 100191, China}

\author{Lin Chen}\email[]{linchen@buaa.edu.cn (corresponding author)}
\affiliation{School of Mathematics and Systems Science, Beihang University, Beijing 100191, China}
\affiliation{International Research Institute for Multidisciplinary Science, Beihang University, Beijing 100191, China}

\begin{abstract}
The entangling power of a bipartite unitary operation shows the maximum created entanglement with the product input states. For an arbitrary two-qubit unitary operation, it is sufficient to consider its normalized operation $U$ with parameters $c_0,c_1,c_2$ and $c_3$. We show how to compute the entangling power of $U$ when $c_2=c_3$. In particular we construct the analytical expressions of entangling power of such $U$ for two examples. We also formulate the entangling power of bipartite unitary operations of Schmidt rank two for any dimensions.
\end{abstract}

\date{\today}

\maketitle

\tableofcontents

\section{Introduction}
\label{sec:intro}
    
In quantum information, bipartite nonlocal unitary gates can create quantum entanglement. They are respectively the fundamental operation and physical resource in quantum-information tasks, such as teleportation \cite{bbc93}, quantum cryptography \cite{ea91} and steering \cite{zmc16}. It is known that the entanglement of a bipartite pure state $\ket{\ps}_{AB}$ can be measured by the von Neumann entropy $S(\cdot)$ of the reduced density matrix on any one system, i.e. $
E(\ket{\psi}_{AB}):=S(\tr_{A} \proj{\psi}):=S(\tr_{B} \proj{\psi})$. Understanding how much entanglement $U$ can create not only characterizes how useful $U$ is in the above quantum-information tasks, but also is related to the reversibility of resources in quantum computation, i.e., whether there is a bipartite unitary operation whose entanglement cost is strictly greater than its ability to create entanglement \cite{lsw09,sm10,mkz13,cy16,lcly20160808}. 

In this paper we investigate the maximum amount of entanglement a bipartite unitary $U$ on the Hilbert space $\cH_A\ox\cH_B$ can create using a product input state \cite{Nielsen03}. The amount is called the entangling power of $U$. Mathematically it is defined as 
\begin{equation}
\label{def:entanglingpower}
\mathit{K}_{\mathit{E}}(U)\coloneqq \max_{\ket{\phi},\ket{\psi}} \mathit{E}(U(\ket{\phi}\ket{\psi})).
\end{equation}
Here $\ket{\phi}$ and $\ket{\psi}$ are respectively pure states on systems $A R_A$ and $B R_B$, $R_A$ and $R_B$ are local reference systems, and $U\ket{\phi,\psi}$ is a bipartite pure state on the system $AR_A$ and $BR_B$. We shall name the "critical state" for $U$ as the states $\ket{\phi,\psi}$ that realize the maximum in Eq. \eqref{def:entanglingpower}. We will investigate two families of $U$. The first family acts on the two-qubit system, and the second family acts on high dimensions and has Schmidt rank two (The Schmidt rank is defined as the smallest number of tensor product matrices \cite{cy15} ). On the one hand, the two-qubit system is one of the fundamental systems in quantum information and has been widely investigated in entanglement measure and physical models \cite{wootters1998,kc2000,zj03,ra04,zj04,vd04,svv04,vfwc04,kb10,ydy13}. On the other hand, it is known that the Schdmit-rank-two bipartite unitary operation is a controlled unitary operation controlled from both the $A$ and $B$ side \cite{cy15,cy13}. Such operations including the known controlled-NOT(CNOT) gates have been extensively investigated in theory and experiments in the past decades \cite{zj04,vd04,svv04,vfwc04,zj040416}. 
As far as we know, it is a hard problem to show the explicit expression of entangling power of a two-qubit unitary operation of Schmidt rank four due to mathematical difficulties, and even showing a non-trival, computationally efficient upper bound of entangling power is very difficult. Recently, Ref. \cite{sd1803} determines a non-trival upper bound of entangling power of bipartite unitary operation. 

We shall present the normalized operation $U$ of an arbitrary two-qubit unitary operation via four parameters $c_0,c_1,c_2$ and $c_3$, defined in Eq. \eqref{eq:twoqubit}. The normalized $U$ is equivalent to the original two-qubit unitary operation when considering entangling power, since it is known that local unitary transformations don't change the entanglement. Hence, it suffices to investigate the entangling power of the normalized $U$. It follows from Eq. \eqref{eq:twoqubit} that the Schmidt rank of $U$ is the number of nonzero coefficients $c_0,c_1,c_2$ and $c_3$. We study Schmidt-rank-four $U$ with $c_2=c_3$ not only because we devise an efficient method to compute its entangling power, but also because such $U$ is of wide interest in quantum information. For example, the family of Schmidt-rank-four $U$ with $c_2=c_3$ includes some essential quantum gates such as SWAP gate and CNOT gate. Studying such $U$ will help us understand how the entangling power of $U$ varies, apart from the SWAP and controlled phase gates. We construct a bipartite unitary $V(\g)$ which commutes with Schmidt-rank-four $U$ with $c_2=c_3$. Using the commutativity, we simplify the critical states as $(\cos\a\ket{00}+\sin\a\ket{11})\otimes(\cos\b\ket{00}+\sin\b\ket{11})$  in Theorem \ref{pp:fam2-1}. We further derive a necessary condition that $\a+\b={\pi\over2}$ for the critical states when $\a\in(0,{\pi\over2})$ and $\b\in(0,{\pi\over2})$. In Proposition \ref{pp:lowerbound} we show the created entanglement is the function of a single variable $\a\in[0,{\pi\over4}]$. Then we analytically derive the expressions of entangling power of two examples, respectively in Theorem \ref{thm:keuofspeu} and Theorem \ref{thm:keuoffam-2}. For the first example, we show the critical states should be either the local product state $(\a=0)$ or the local maximally entangled state $(\a={\pi\over4})$, up to equivalence. For the second example, we show the critical states should only be the local maximally entangled state $(\a={\pi\over4})$, up to equivalence. We also formulate the entangling power of Schmidt-rank-two bipartite unitary operations for dimensions higher than $3$ in Theorem \ref{thm:Schr2}.

The rest of this paper is organized as follows. In Sec. \ref{sec:pre} we present some preliminary results used in this paper. In Sec. \ref{sec:normaldecom}
we characterize the normal decomposition and mathematical properties of two-qubit unitary operation $U$. In Sec. \ref{sec:2qubit} we first introduce the physical meaning of a family of two-qubit unitary operations $U$ defined in Eq. \eqref{eq:twoqubit} with $c_2=c_3$ as our research motivation, and then investigate the entangling power of $U$ in this family. We construct the analytical forms of entangling power for two examples of such $U$ in Sec. \ref{sec:specialU}. For higher dimensions, we formulate the entangling power of Schmidt-rank-two bipartite unitary gates in Sec. \ref{sec:Schr2}. Finally we present open problems in Sec. \ref{sec:open}.

\section{Preliminaries}
\label{sec:pre}

In this section we introduce two preliminary results from quantum information. They will be used to support our main results. For two density matrices $\r$ and $\s$, we write $\r\prec\s$ to denote that the eigenvalue vector of $\r$ is majorized by that of $\s$. The following fact is known.
\begin{lemma}\cite{mc2011}
\label{le:maj}
If $\r\prec\s$ then $S(\r)\ge S(\s)$.	
\end{lemma}

\begin{lemma}\cite[Theorem 11.9]{nc2000book}
\label{le:povm}
Let $\{P_i\}$ be an orthogonal complete POVM, i.e., $\sum_i P_i = I$ and $P_i P_j = \d_{ij}P_i$. Let $\r$ be a quantum state. Then
\begin{eqnarray}
S( \sum_i P_i \r P_i ) \ge S(\r),	
\end{eqnarray}
where the equality holds iff $\sum_i P_i \r P_i = \r$.	
\end{lemma}

A natural measure of entanglement for $U$ is the Schmidt Strength $K_{Sch}(U)$ defined as follows. We can always decompose $U$ acting on systems $A$ and $B$ as
$
U=\sum_j s_j A_j\otimes B_j,
$
where $s_j\geq0$ and $\{A_j\}$ and $\{B_j\}$ are orthonormal operator bases for $A$ and $B$ whose dimensions are $d_A$ and $d_B$, respectively. Then $K_{Sch}(U)$ is defined to be the Shannon entropy $H(\cdot)$ of the distribution $s_j^2/(d_A d_B)$,
\begin{equation}
\label{eq:kschU}
K_{Sch}(U):=H\big(\big\{\frac{s_j^2}{d_A d_B}\big\}\big).
\end{equation}
$K_{Sch}(U)$ has several good properties and is related to $K_E(U)$. For example, $K_{Sch}(U)$ satisfies the properties of exchange symmetry, time-reversal invariance, stability with respect to local ancillas, and is strongly additive distinctively.
It is shown that $K_{Sch}(U)$ is a lower bound of $K_E(U)$ for all unitaries $U$, and is a strict lower bound for $U_p=(\sqrt{1-p}\s_0 \otimes \s_0+i\sqrt{p} \s_1\otimes \s_1)(\sqrt{1-p}\s_0 \otimes \s_0+i\sqrt{p} \s_3\otimes \s_3)$ with certain $p$ \cite[Theorem 1]{Nielsen03}.

\section{Normalized two-qubit unitary operations $U$}
\label{sec:normaldecom}

In this section we characterize the normal decomposition of an arbitrary two-qubit unitary operation and mathematical properties of the normalized operation $U$.
In Sec. \ref{subsec:decom}, we introduce a normal decomposition in Eq. \eqref{eq:decom}. The $U$ in terms of three real parameters $x,y,z$ in Eq. \eqref{eq:twoqubit} is called normalized. 
For simplifying the computation of entangling power of $U$, we present constraints on $x,y,z$ in \eqref{eq:xyz} and \eqref{eq:xyz2}. In Sec. \ref{subsec:mathproperty} we present some widely used mathematical properties of the normalized $U$ satisfying 
\eqref{eq:xyz} and \eqref{eq:xyz2}.

\subsection{Normal decomposition}
\label{subsec:decom}

The number of parameters of any two-qubit unitary operation $W_{AB}$ can be reduced from $15$ to $3$ through an efficient decomposition \cite{kc2000}. That is, there exists local unitary operators $U_A,U_B,V_A,V_B$ and the two-qubit unitary operator $U$ such that 
\begin{eqnarray}
\label{eq:decom}
W_{AB}=(U_A\otimes U_B) U (V_A\otimes V_B),
\end{eqnarray}
where
\begin{eqnarray}
\label{eq:twoqubit}	
U:=\sum^3_{j=0} c_j\s_j \ox \s_j.
\end{eqnarray}
Here $\s_j$'s are the Pauli matrices, i.e., 
\begin{eqnarray}
\label{eq:s1}
\s_0=\left(
                   \begin{array}{cc}
                     1 & 0 \\
                     0 & 1 \\
                     \end{array}
                 \right),
~~
\s_1=\left(
                   \begin{array}{cc}
                     0 & 1 \\
                     1 & 0 \\
                     \end{array}
                 \right),
~~
\s_2=\left(
                   \begin{array}{cc}
                     0 & -i \\
                     i & 0 \\
                     \end{array}
                 \right),
~~
\s_3=\left(
                   \begin{array}{cc}
                     1 & 0 \\
                     0 & -1 \\
                     \end{array}
                 \right),
\end{eqnarray}
and the coefficients $c_j$'s are complex numbers   \cite[Eq. (4.3)]{Nielsen03}
\begin{equation}
\label{eq:c0-c3exp}
\begin{aligned}
c_0&=\cos x\cos y\cos z+i \sin x\sin y\sin z,
\\
c_1&=\cos x\sin y\sin z+i \sin x\cos y\cos z,
\\
c_2&=\sin x\cos y\sin z+i \cos x\sin y\cos z,
\\
c_3&=\sin x\sin y\cos z+i \cos x\cos y\sin z,	
\end{aligned}
\end{equation}
for $x,y,z\in(-\p/4,\p/4]$.
As we know, the entanglement is not changed by local unitary transformations for any measure. So for any two-qubit unitary operations it suffices to study the entangling power $K_E(U)$ of the normalized $U$.

\begin{remark}
\label{re:2qubitu-1}
Since $\s_j$'s are linearly independent, the Schmidt rank of $U$ is the same as the number of nonzero $c_j$'s in \eqref{eq:twoqubit}. If some $c_j$ is zero then the entangling power of $U$ is known \cite[Proposition 2. and Theorem 2.]{Nielsen03}, and we have extended it to arbitrary Schmidt-rank-two bipartite unitaries in Sec. \ref{sec:Schr2}.  The entangling power of any bipartite permutation unitary of Schmidt rank three is also known from \cite[Proposition 3]{lcly20160808}. Further, if all $c_j$'s have modulus $1/2$ then it is easy to verify that the entangling power of $U$ reaches the maximum, i.e., $2$ ebits.
\end{remark}

The remaining problem is how to obtain the entangling power of $U$ of Schmidt rank four. Therefore, we shall assume that $c_j\ne0$ for any $j$, and one of $c_j$'s has modulus greater than $1/2$. For convenience, we can impose restrictions on parameters $x,y,z$ as follows.
\begin{eqnarray}
\label{eq:xyz}	
&&\frac{\pi}{4}\geq x\geq y\geq z\geq 0,\\
\label{eq:xyz2}
&&{\p\over4}>y>0.
\end{eqnarray}
Eq. \eqref{eq:xyz} follows from \cite[Eq. (13)]{kc2000}. For the entangling power of $U$ whose $x,y,z$ don't satisfy \eqref{eq:xyz} is equal to some satisfying \eqref{eq:xyz} via some symmetric relations \cite{kc2000}, and Eq. \eqref{eq:xyz2} follows from $c_j\ne0$ for any $j$, and one of $c_j$'s has modulus greater than $1/2$. Hence, we will investigate the entangling power of the normalized $U$ satisfying \eqref{eq:xyz} and \eqref{eq:xyz2}.

\subsection{Mathematical properties}
\label{subsec:mathproperty}

From \eqref{eq:c0-c3exp}, we find $c_j$'s are closely related from each other. We construct a few formulas on $c_j$'s. They will be widely used in the remaining sections. 
\begin{eqnarray}
&&\label{eq:c03+}
\abs{c_0+c_3}^2-\abs{c_1-c_2}^2=\cos2(x-y),
\\&&\label{eq:c03+}
\abs{c_0-c_3}^2-\abs{c_1+c_2}^2=\cos2(x+y),
\\&&\label{eq:c03}
\abs{c_0+c_3}^2+\abs{c_1-c_2}^2=\abs{c_0-c_3}^2+\abs{c_1+c_2}^2=1,
\\&&\label{eq:c0}
\abs{c_0}\ge \max_{j=1,2,3}\{\abs{c_j}\},
\\&&\label{eq:c0>12}
\abs{c_0}>{1\over2},
\\&&\label{eq:cos2x}
\abs{c_0}^2+\abs{c_3}^2
=
{1\over2}
(1+\cos2x\cos2y)
>{1\over2}
> 
\abs{c_1}^2+\abs{c_2}^2,
\\&&\label{eq:c0c3c1c2}
c_0c_3^*+c_0^*c_3=
c_1c_2^*+c_1^*c_2=
{1\over2}\sin2x\sin2y>0,
\\&&\label{eq:c0-c3}
(c_0c_3^*-c_0^*c_3)^2=-\frac{1}{4}(\cos2x+\cos2y)^2\sin^22z\leq 0,
\\&&\label{eq:c1-c2}
(c_1c_2^*-c_1^*c_2)^2=-\frac{1}{4}(\cos2x-\cos2y)^2\sin^22z\leq 0,
\\&&\label{eq:c0c3}
\abs{c_0c_3}^2-\abs{c_1c_2}^2
={1\over4}\cos2x\cos2y(\sin2z)^2\ge0,
\\&&\label{eq:c02-c12}
\abs{c_0}^2-\abs{c_1}^2
={1\over2}\cos2x(\cos2y+\cos2z),
\\&&\label{eq:c02-c22}
\abs{c_0}^2-\abs{c_2}^2
={1\over2}\cos2y(\cos2x+\cos2z),
\\&&\label{eq:c02-c32}
\abs{c_0}^2-\abs{c_3}^2
={1\over2}\cos2z(\cos2x+\cos2y).
\end{eqnarray}
Eq. \eqref{eq:c0>12} is from the assumption that one of $c_j$'s has modulus greater than $1/2$, and the restriction \eqref{eq:xyz2} assures the positivity of Eq. \eqref{eq:c0c3c1c2}. Eqs. \eqref{eq:c02-c12}-\eqref{eq:c02-c32} imply that if $\abs{c_0}=\abs{c_i}$ for some $i\in\{1,2,3\}$ then $\abs{c_j}=\abs{c_k}$ for $j,k\in\{1,2,3\}\setminus \{i\}$, and the same equations imply that if $\abs{c_1}=\abs{c_2}=\abs{c_3}$ then they are also equal to $\abs{c_0}$.

By the definition of entangling power in Eq. \eqref{def:entanglingpower}, we select a family of input states defined as follows. Up to equivalence, one can show that they cover all input states for computing the entangling power of $U$.

\begin{definition}
\label{def:inputstates}
\begin{equation}
\label{eq:definstates1}
\begin{aligned}
\ket{\ps(\a;\t,\m)}=\cos\a\ket{0,0}+\sin\a\ket{1}(e^{i\t}\cos \m\ket{0}+\sin \m\ket{1}),\\
\ket{\phi(\b;\x,\n)}=\cos\b\ket{0,0}+\sin\b\ket{1}(e^{i\x}\cos \n\ket{0}+\sin \n\ket{1}),\\
\end{aligned}
\end{equation}
where
\begin{eqnarray}
\label{eq:ab}
\a\in[0,\frac{\pi}{2}],
~~~~~~
\b\in[0,\frac{\pi}{2}],
~~~~~~
\t,\x\in[0,2\p),
~~~~~~
\m,\n\in(0,\frac{\pi}{2}].
\end{eqnarray}
$\ket{\ps(\a;\t,\m)}$ is the input state on system $AR_A$ and $\ket{\phi(\b;\x,\n)}$ is the input state on system $BR_B$.
Then the output state of $U$ is the bipartite pure state
\begin{eqnarray}
\label{eq:upsi}	
\ket{\ph^\t_\x(\a,\b;\m,\n)}:=U(\ket{\ps(\a;\t,\m)}_{AR_A}\ox\ket{\phi(\b;\x,\n)}_{BR_B})
\end{eqnarray}
which has entanglement $E(\ph^\t_\x(\a,\b;\m,\n))$.
\qed
\end{definition}

\begin{remark}
\label{re:defofinput-1}
Since $E(\ph^\t_\x(\a,\b;\m,\n))$ is a continuous function with $\a,\b,\t,\x,\m,\n$, $E(\ph^\t_\x(\a,\b;\m,\n))$ reaches its maximum in its domain \eqref{eq:ab}. The maximum is equal to the entangling power of $U$ by definition. We shall omit $\t$ and $\x$ respectively when we set $\m=\frac{\pi}{2}$ and $\n=\frac{\pi}{2}$ respectively.
\end{remark}

Though computing the entangling power of $U$ of Schmidt rank four in \eqref{eq:twoqubit} is a hard problem, we will show a practical method of computing the entangling power for the family of $U$ with $c_2=c_3$ in the next section by using the normal decomposition and mathematical properties of $U$ developed in this section.

\section{The physical meaning of $U$ with $c_2=c_3$ and its entangling power}
\label{sec:2qubit}

We begin this section with introducing the rich physical meaning of normalized $U=\sum^3_{j=0} c_j\s_j\ox\s_j$ with $c_2=c_3$ in \eqref{eq:c0-c3exp}. It is our motivation to investigate the entangling power of such $U$.

We stress that the normalized $U$ with $c_2=c_3$ in \eqref{eq:c0-c3exp} is of wide interest in quantum information.
The condition $c_2=c_3$ is equivalent to $y=z$, hence
\begin{eqnarray}
\label{eq:y=z}
U=
\bma 
e^{iy}\cos(x-y)&0&0&ie^{iy}\sin(x-y)\\
0&e^{-iy}\cos(x+y)&ie^{-iy}\sin(x+y)&0\\
0&ie^{-iy}\sin(x+y)&e^{-iy}\cos(x+y)&0\\
e^{iy}\cos(x-y)&0&0&ie^{iy}\sin(x-y)\\
\ema.	
\end{eqnarray}
First if $x=y={{\p\over4}}$ then $U$ becomes the well-known SWAP gate. It has the maximum entangling power 2 ebits, and its implementation requires  three controlled unitary gates and some local unitary gates \cite[Lemma 1]{cy15}. Second, if $y=0$ then up to local unitary gates $U$ is equivalent to the diagonal matrix $\diag(1,1,1,e^{4ix})$. Since $x\in\bbR$, this is an arbitrary controlled phase gate. For example the controlled-NOT gate is a basic ingredient in quantum circuit, and it is known that the CNOT gate has entangling power 1 ebit. It is thus a more general and fundmental question to understand how the entangling power of $U$ varies, apart from the SWAP and controlled phase gates. Third,
there is an interesting subfamily of $U$ with $c_2=c_3$, that is 
\begin{equation}
\label{eq:U_p-1}
\begin{aligned}
U_p&=(1-p)\s_0 \otimes \s_0+ p\s_1 \otimes \s_1+ i\sqrt{p(1-p)}(\s_2 \otimes \s_2+\s_3 \otimes \s_3)\\
   &=\big(\sqrt{1-p}\s_0 \otimes \s_0+i\sqrt{p} \s_2\otimes \s_2)(\sqrt{1-p}\s_0 \otimes \s_0+i\sqrt{p} \s_3\otimes \s_3 \big).
\end{aligned}
\end{equation}
Firstly, $U_p$ can be decomposed into a product of two Schmidt-rank-two unitary operations, so it is more realizable in experiments. Secondly, its properties about entanling power is clear. For example, it can be used to construct unitaries $U$ such that $K_E^{A_1A_2:B_1B_2}(U_{A_1B_1}\otimes U_{A_2B_2})>2K_E^{A_1B_1}(U_{A_1B_1})$, where the subscripts on $U$ indicate the subsystems to which it is applied \cite[Theorem 3]{Nielsen03}. It implies $K_E(U)$ is superadditive, while the Schmidt Strength $K_{Sch}(U)$ is additive.
Fourth, the two-qubit X state 
$
\r=
\bma  
a_1 &0&0&a_2\\
0&b_1&b_2&0\\
0&b_2^*&b_3&0\\
a_2^* &0&0&a_3\\
\ema
$ has been extensively studied in the past years, not only for entanglement, but also for other kinds of quantum correlations, such as quantum discord \cite{MAG10,CQXstates11}. Quantum discord is defined as the difference between quantum mutual information and classical correlation in a bipartite system, and has been shown it can offer some advantage to several tasks in quantum information processing. We claim that any $\r$ can be realized by performing some $U$ and local unitary $D$ on a diagonal two-qubit state $\s$. The latter is a classical-correlated state containing no quantum correlation, namely it has zero quantum discord \cite{LV01,VVcorrelations03}. To explain our claim, we note that $\r=V\s V^\dg$ with an $X$ unitary operation $V=\bma 
\cos(x-y)&0&0&\sin(x-y)\\
0&\cos(x+y)&\sin(x+y)&0\\
0&e^{im}\sin(x+y)&-e^{im}\cos(x+y)&0\\
e^{in}\cos(x-y)&0&0&-e^{in}\sin(x-y)\\
\ema$ for real $m,n$,  
and a diagonal $\s$. Let $D=\diag(1,e^{ip})\ox \diag(1,e^{iq})$ and $F$ be a diagonal unitary gate and a local unitary gate, respectively. By choosing suitable $D$ and $F$ one can obtain that $U=DVF$. So $\r=D^\dg U\s U^\dg D$, and we have proven our claim. One can extend the above claim to constructing multiqubit $X$ states using more $U$'s in \eqref{eq:y=z}, and such states have recently been interesting in entanglement theory \cite{Chen2017Separability, Han2016Construction, Mendonca2015Maximally}. The claim is relatively realizable in experiments, because $D$ and $\s$ are both diagonal, and $U$ is of an $X$-type operation. By investigating the entangling power of $U$ with $c_2=c_3$, we can see how the entanglement of two-qubit $X$ state varies, since $\s$ contains no entanglement.

In the following of this section, we investigate the entangling power of normalized $U=\sum^3_{j=0} c_j\s_j\ox\s_j$ with $c_2=c_3$ in \eqref{eq:c0-c3exp}. We will show that the critical states of such $U$ should be $\ket{\ph(\a,\b;\frac{\pi}{2},\frac{\pi}{2})}$ defined in \eqref{eq:upsi}. Then $E(\ph(\a,\b;\frac{\pi}{2},\frac{\pi}{2}))$ is the entanglement produced by $\ket{\ph(\a,\b;\frac{\pi}{2},\frac{\pi}{2})}$.

It has been shown that for normalized $U$ satisfying the condition $c_1=c_2=c_3$, i.e., $x=y=z$ in \eqref{eq:c0-c3exp}, the critical states can always be written as $\ket{\ph(\a,\b;\frac{\pi}{2},\frac{\pi}{2})}$ no matter which measure is selected \cite{kc2000}. In Theorem \ref{pp:fam2-1} we generalize the result to the case when $U$ satisfy the condition $c_2=c_3$ i.e., $y=z$ if we select the von Neumann entropy as the measure. Before that we present Lemma \ref{le:nu=pi/2} to support Theorem \ref{pp:fam2-1}.

\begin{lemma}
\label{le:nu=pi/2}	
The critical states of $U$ can be written as $\ket{\ph(\a,\b;\frac{\pi}{2},\frac{\pi}{2})}$ if they  can be written as $\ket{\ph^\t(\a,\b;\m,\frac{\pi}{2})}$.
\end{lemma}
\begin{proof}
Suppose $\rho_{BR_B}$ represents the reduced density matrix and $\rho_{BR_B}'=P_1\rho_{BR_B}P_1^\dg+P_2\rho_{BR_B}P_2^\dg$ where $P_1=\ketbra{00}{00}+\ketbra{11}{11}$ and $P_2=\ketbra{01}{01}+\ketbra{10}{10}$. Since $\{P_1,P_2\}$ is an orthogonal complete POVM, Lemma \ref{le:povm} implies $S(\rho_{BR_B})\leq S(\rho_{BR_B}')$. 

Suppose $\r_{BR_B}$ and $\r_{BR_B}'$ respectively become $\rho_{1BR_B}$ and $\rho_{1BR_B}'$ when we set the input states as $\ket{\ph^\t(\a,\b;\m,\frac{\pi}{2})}$, and respectively become $\rho_{2BR_B}$ and $\rho_{2BR_B}'$ when we set the input states as $\ket{\ph(\a,\b;\frac{\pi}{2},\frac{\pi}{2})}$. If the critical states can be written as $\ket{\ph^\t(\a,\b;\m,\frac{\pi}{2})}$, then the maximum von Neumann entropy of $\rho_{1BR_B}'$ is an upper bound of the entangling power. Next we will show this upper bound can be attained via the input states $\ket{\ph(\a,\b;\frac{\pi}{2},\frac{\pi}{2})}$. One can verify that $\rho_{1BR_B}'=\rho_{2BR_B}'=\rho_{2BR_B}$. So we have $K_E(U)\leq\max S(\rho_{1BR_B}')=\max S(\rho_{2BR_B}')=\max S(\rho_{2BR_B})\leq K_E(U)$. It implies $\max S(\rho_{2BR_B})= K_E(U)$. So we conclude the critical states can always be written as $\ket{\ph(\a,\b;\frac{\pi}{2},\frac{\pi}{2})}$.
\end{proof}

\begin{theorem}
\label{pp:fam2-1}
If normalized operations $U$ satisfy $c_2=c_3$, then $K_E(U)=\max\limits_{\a,\b\in[0,\p/2]}E(\ph(\a,\b;\frac{\pi}{2},\frac{\pi}{2}))$ with the critical states as $\ket{\ph(\a,\b;\frac{\pi}{2},{\p\over2})}
$.
\end{theorem}

\begin{proof}
We first propose the following observations and one can verify them straightforwardly.
\begin{eqnarray}
\label{eq:commute}
&&(\s_j\ox \s_k+\s_k\ox \s_j)
(\s_j\ox \s_j+\s_k\ox \s_k)=0,
\\&& 
j,k=1,2,3,
\quad
j\ne k.
\notag
\end{eqnarray}

When $c_2=c_3$, using \eqref{eq:commute} we can show that 
\begin{eqnarray}
[V(\g)
\ox
V(\g),
U]=0,
\end{eqnarray}
where $V(\g)=\cos\g\s_2+\sin\g\s_3$ is unitary and $\g$ is an arbitrary real number.
Using \eqref{eq:upsi}, we obtain
\begin{eqnarray}
\label{eq:upsi2}	
&&
(V(\g)
\ox
V(\g))_{AB}\ket{\ph^\t_\x(\a,\b;\m,\n)}\notag
\\=&&
U
\bigg(
(V(\g)\ox I_{R_A})\ket{\ps(\a;\t,\m)}_{AR_A}
\ox
(V(\g)\ox I_{R_B})
\ket{\phi(\b;\x,\n)}_{BR_B}
\bigg).
\end{eqnarray}
Since performing local unitaries on $R_A,R_B$ does not change the entanglement of $\ket{\ph^\t_\x(\a,\b;\m,\n)}$, we may assume that $(V(\g_1)\ox W_1)\ket{\ps(\a;\t,\m)}=\cos\a\ket{00}+\sin\a\ket{11}$ by choosing a suitable $\g_1$ and unitary $W_1$, or we may assume that $(V(\g_2)\ox W_2)\ket{\phi(\b;\x,\n)}=\cos\b\ket{00}+\sin\b\ket{11}$ by choosing another suitable $\g_2$ and unitary $W_2$. Hence, Lemma \ref{le:nu=pi/2} implies we can choose the critical state for computing $K_E(U)$ as 
\begin{eqnarray}
\label{eq:inputspe31}
\ket{\ph(\a,\b;\frac{\pi}{2},{\p\over2})}
=
(\cos\a\ket{00}+\sin\a\ket{11})_{AR_A}
\ox		
(\cos\b\ket{00}+\sin\b\ket{11})_{BR_B},
\end{eqnarray}
for $\a,\b\in[0,\p/2]$.
By definition we have $K_E(U)=\max\limits_{\a,\b\in[0,\p/2]}E(\ph(\a,\b;\frac{\pi}{2},\frac{\pi}{2}))$.

This completes the proof.
\end{proof}

In the following part of this section we will determine $\a$ and $\b$ such that $E(\ph(\a,\b;\frac{\pi}{2},\frac{\pi}{2}))$ reaches its maximum.

By computing, we have the reduced density matrix $\rho_{BR_B}$ as follows.
\begin{equation}
\label{eq:rhobrb}
\begin{aligned}
\rho_{BR_B}&:=\tr_{AR_A}\Big[U(\ket{\ps(\a;\t,\m)}_{AR_A}\ox\ket{\phi(\b;\x,\n)}_{BR_B})\\
&(\bra{\ps(\a;\t,\m)}_{AR_A}\ox\bra{\phi(\b;\x,\n)}_{BR_B})U^\dg\Big]\\
&=P^\dg
\bma
m_{11}(\a,\b) & m_{14}(\a,\b) & 0 & 0 \\
m_{41}(\a,\b) & m_{44}(\a,\b) & 0 & 0 \\
0 & 0 & m_{22}(\a,\b) & m_{23}(\a,\b) \\
0 & 0 & m_{32}(\a,\b) & m_{33}(\a,\b) \\
\ema
P,
\end{aligned}
\end{equation}
where $P$ is a permutation matrix and
\begin{equation}
\label{eq:Ephialphabetax}
\begin{aligned}
m_{11}(\a,\b)&=\cos^2\b\big(\abs{c_0}^2+\abs{c_3}^2+\cos2\a(c_0c_3^*+c_3c_0^*)\big), \\
m_{14}(\a,\b)&=\frac{1}{2}\sin2\b\big(\abs{c_0}^2-\abs{c_3}^2+\cos2\a(-c_0c_3^*+c_3c_0^*)\big), \\
m_{22}(\a,\b)&=\cos^2\b\big(\abs{c_1}^2+\abs{c_2}^2+\cos2\a(-c_1c_2^*-c_2c_1^*)\big), \\
m_{23}(\a,\b)&=\frac{1}{2}\sin2\b\big(\abs{c_1}^2-\abs{c_2}^2+\cos2\a(c_1c_2^*-c_2c_1^*)\big), \\
m_{32}(\a,\b)&=\frac{1}{2}\sin2\b\big(\abs{c_1}^2-\abs{c_2}^2+\cos2\a(-c_1c_2^*+c_2c_1^*)\big), \\
m_{33}(\a,\b)&=\sin^2\b\big(\abs{c_1}^2+\abs{c_2}^2+\cos2\a(c_1c_2^*+c_2c_1^*)\big), \\
m_{41}(\a,\b)&=\frac{1}{2}\sin2\b\big(\abs{c_0}^2-\abs{c_3}^2+\cos2\a(c_0c_3^*-c_3c_0^*)\big), \\
m_{44}(\a,\b)&=\sin^2\b\big(\abs{c_0}^2+\abs{c_3}^2+\cos2\a(-c_0c_3^*-c_3c_0^*)\big).
\end{aligned}
\end{equation}
One can show the eigenvalues of $\rho_{BR_B}$ are the following nonnegative numbers
\begin{equation}
\label{eq:Ephialphabetaeigen}
\begin{aligned}
\lambda_1(\a,\b)
         &=\frac{t_1-\sqrt{t_1^2-4\abs{c_0c_3}^2\sin^22\a\sin^22\b}}{2},\\
\lambda_2(\a,\b)
         &=\frac{t_1+\sqrt{t_1^2-4\abs{c_0c_3}^2\sin^22\a\sin^22\b}}{2},\\
\lambda_3(\a,\b)
         &=\frac{t_2-\sqrt{t_2^2-4\abs{c_1c_2}^2\sin^22\a\sin^22\b}}{2},\\
\lambda_4(\a,\b)
         &=\frac{t_2+\sqrt{t_2^2-4\abs{c_1c_2}^2\sin^22\a\sin^22\b}}{2},    
\end{aligned}
\end{equation}
where 
\begin{equation}
\label{eq:t1t2}
\begin{aligned}
t_1&:=&\abs{c_0}^2+\abs{c_3}^2+\cos2\a\cos2\b(c_0c_3^*+c_3c_0^*),
\\
t_2&:=&\abs{c_1}^2+\abs{c_2}^2-\cos2\a\cos2\b(c_1c_2^*+c_2c_1^*).	
\end{aligned}
\end{equation}
By definition we have 
\begin{equation}
\label{eq:entropyentanglement}
E(\ph(\a,\b;\frac{\pi}{2},\frac{\pi}{2}))=H(\lambda_1(\a,\b),\lambda_2(\a,\b),\lambda_3(\a,\b),\lambda_4(\a,\b)).
\end{equation}

The following lemma gives the maximum of $E(\ph(\a,\b;\frac{\pi}{2},\frac{\pi}{2}))$ when points $(\a,\b)$ on the boundary.

\begin{lemma}
\label{le:boundarycases}
(i)
\begin{eqnarray}
\label{eq:max}
\max_{\a,\b\in[0,\p/2]}	
E(\ph(\a,\b;\frac{\pi}{2},\frac{\pi}{2}))
=&&
\max_{\a\in[0,\p/4],\b\in[0,\p/2]}	
E(\ph(\a,\b;\frac{\pi}{2},\frac{\pi}{2})),
\\
=&&
\max_{\a\in[0,\p/2],\b\in[0,\p/4]}	
E(\ph(\a,\b;\frac{\pi}{2},\frac{\pi}{2})).
\end{eqnarray}
From now on we shall work with the interval $\a\in[0,\p/4],\b\in[0,\p/2]$.

(ii) Suppose set $B=\big\{(\a,\b)|\text{$\a$ or $\b$ $\in$ \{$0$,${\p\over4}$\}}\big\}$. Then we have 
\begin{equation}
\label{eq:maxonboundary}
\max_{(\a,\b)\in B}E(\ph(\a,\b;\frac{\pi}{2},\frac{\pi}{2}))=
\begin{cases}
\max\big\{1,E(\ph(\frac{\pi}{4},\frac{\pi}{4};\frac{\pi}{2},\frac{\pi}{2}))\big\},& \text{$\cos(2x+2y)\leq 0$},\\
\max\big\{E(\ph(0,\frac{\pi}{2};\frac{\pi}{2},\frac{\pi}{2})),E(\ph(\frac{\pi}{4},\frac{\pi}{4};\frac{\pi}{2},\frac{\pi}{2}))\big\},& \text{$\cos(2x+2y)> 0$},
\end{cases}
\end{equation}
where $x,y$ are in \eqref{eq:c0-c3exp}.
\end{lemma}

The proof of Lemma \ref{le:boundarycases} has been shown in Appendix \ref{sec:proofoflemma4}. Lemma \ref{le:boundarycases} (i) shows it suffices to consider the case $\a\in[0,\frac{\pi}{4}],\b\in[0,\frac{\pi}{2}]$, and (ii) gives the maximum on the boundary. It remains to find out the maximum of $E(\ph(\a,\b;\frac{\pi}{2},\frac{\pi}{2}))$ when $\a\in(0,{\p\over4}),\b\in(0,{\p\over2})$. It is necessary to find out the extreme points. It is known that a necessary condition for the extreme points $(\a,\b)$ of $E(\ph(\a,\b;\frac{\pi}{2},\frac{\pi}{2}))$ is that they make each partial derivative of $E(\ph(\a,\b;\frac{\pi}{2},\frac{\pi}{2}))$ equal zero. One can formulate the two partial derivatives as follows.

\begin{equation}
\label{eq:partialderivatives}
\begin{aligned}
{\partial E(\ph(\a,\b;\frac{\pi}{2},\frac{\pi}{2}))
\over
\partial \a} :=
f_{\a}&=k\sin2\a\cos2\b\cdot g(\a,\b)+\sin2\a\cos2\a\sin^22\b\cdot h(\a,\b),
\\
{\partial E(\ph(\a,\b;\frac{\pi}{2},\frac{\pi}{2}))
\over
\partial \b} :=
f_{\b}&=k\cos2\a\sin2\b\cdot g(\a,\b)+\sin2\b\cos2\b\sin^22\a\cdot h(\a,\b),
\end{aligned}
\end{equation}
where 
\begin{equation}
\label{eq:partialderivatives1}
\begin{aligned}
g(\a,\b)&=\log\frac{l_1}{l_2}+\frac{t_1}{\sqrt{t_1^2-4l_1\sin^22\a\sin^22\b}}\log\frac{\lambda_2}{\lambda_1}-\frac{t_2}{\sqrt{t_2^2-4l_2\sin^22\a\sin^22\b}}\log\frac{\lambda_4}{\lambda_3},\\
h(\a,\b)&=\frac{4l_1}{\sqrt{t_1^2-4l_1\sin^22\a\sin^22\b}}\log\frac{\lambda_2}{\lambda_1}+\frac{4l_2}{\sqrt{t_2^2-4l_2\sin^22\a\sin^22\b}}\log\frac{\lambda_4}{\lambda_3}>0,
\end{aligned}
\end{equation}
and the constants
\begin{equation}
\label{eq:eq:partialderivatives11}
\begin{aligned}
k&=c_0c_3^*+c_3c_0^*,\quad b&=\abs{c_0}^2+\abs{c_3}^2,\quad l_1&=\abs{c_0c_3}^2,\quad l_2&=\abs{c_1c_2}^2.
\end{aligned}
\end{equation}
The following lemma presents a necessary condition for the points which realize the maximum of $E(\ph(\a,\b;\frac{\pi}{2},\frac{\pi}{2}))$ when $\a\in(0,\frac{\pi}{4}],\b\in(0,\frac{\pi}{2})$.

\begin{lemma}
\label{le:maxa=b}
Suppose $\a\in(0,\frac{\pi}{4}],\b\in(0,\frac{\pi}{2})$. The maximum of $E(\ph(\a,\b;\frac{\pi}{2},\frac{\pi}{2}))$ occurs only when $\a+\b=\pi/2$, i.e., 
\begin{eqnarray}
\label{eq:necon1}
\max_{\a\in(0,\p/4],\b\in(0,\p/2)}	
E(\ph(\a,\b;\frac{\pi}{2},\frac{\pi}{2}))
=
\max\limits_{\a\in(0,\p/4]}
E(\ph(\a,\frac{\pi}{2}-\a;\frac{\pi}{2},\frac{\pi}{2})).	
\end{eqnarray}
\end{lemma}

\begin{proof}
Set the two equations in Eq. \eqref{eq:partialderivatives} equal zero. When $\a,\b\neq\frac{\pi}{4}$  we have
\begin{equation}
\label{eq:partialderivativesx}
\left\{
\begin{aligned}
\frac{g(\a,\b)}{h(\a,\b)}&=-\frac{\cos2\a\sin^22\b}{k\cos2\b},\\
\frac{g(\a,\b)}{h(\a,\b)}&=-\frac{\cos2\b\sin^22\a}{k\cos2\a}.
\end{aligned}
\right.
\end{equation}
We derive $(\sin(2\a+2\b))(\sin(2\a-2\b))=0$ from \eqref{eq:partialderivativesx}. It implies $\a+\b=\pi/2$ or $\a=\b$. When $\a=\frac{\pi}{4}$ or $\b=\frac{\pi}{4}$, Eq. \eqref{eq:necon1} follows from $\max\limits_{\b\in[0,\frac{\pi}{2}]}E(\ph(\frac{\pi}{4},\b;\frac{\pi}{2},\frac{\pi}{2}))=\max\limits_{\a\in[0,\frac{\pi}{4}]}E(\ph(\a,\frac{\pi}{4};\frac{\pi}{2},\frac{\pi}{2}))=E(\ph(\frac{\pi}{4},\frac{\pi}{4};\frac{\pi}{2},\frac{\pi}{2}))$. We next rule out the necessary condition $\a=\b$. That is we will show $E(\ph(\a,\a;\frac{\pi}{2},\frac{\pi}{2}))\leq E(\ph(\frac{\pi}{4},\frac{\pi}{4};\frac{\pi}{2},\frac{\pi}{2})),\forall \a\in[0,\frac{\pi}{4}]$.

For convenience, when $\b=\a$ and ${\p\over2}-\a$ in \eqref{eq:Ephialphabetaeigen}, we respectively name the eigenvalues as in Eq. \eqref{eq:Ephialphabetaeigen1} and Eq. \eqref{eq:Ephialphabetaeigen2}.
\begin{equation}
\label{eq:Ephialphabetaeigen1}
\begin{aligned}
\lambda_{11}
         &:=\frac{t_{11}-\sqrt{t_{11}^2-4\abs{c_0c_3}^2\sin^42\a}}{2},\\
\lambda_{21}
         &:=\frac{t_{11}+\sqrt{t_{11}^2-4\abs{c_0c_3}^2\sin^42\a}}{2},\\
\lambda_{31}
         &:=\frac{t_{21}-\sqrt{t_{21}^2-4\abs{c_1c_2}^2\sin^42\a}}{2},\\
\lambda_{41}      
         &:=\frac{t_{21}+\sqrt{t_{21}^2-4\abs{c_1c_2}^2\sin^42\a}}{2},
\end{aligned}
\end{equation}
where 
\begin{eqnarray}
\label{eq:t11}
t_{11}&:=&\abs{c_0}^2+\abs{c_3}^2+\cos^22\a(c_0c_3^*+c_3c_0^*),
\\\label{eq:t21}
t_{21}&:=&\abs{c_1}^2+\abs{c_2}^2-\cos^22\a(c_1c_2^*+c_2c_1^*),	
\end{eqnarray}
and
\begin{equation}
\label{eq:Ephialphabetaeigen2}
\begin{aligned}
\lambda_{12}
         &=\frac{t_{12}-\sqrt{t_{12}^2-4\abs{c_0c_3}^2\sin^42\a}}{2},\\
\lambda_{22}
       &=\frac{t_{12}+\sqrt{t_{12}^2-4\abs{c_0c_3}^2\sin^42\a}}{2},\\
\lambda_{32}
         &=\frac{t_{22}-\sqrt{t_{22}^2-4\abs{c_1c_2}^2\sin^42\a}}{2},\\
\lambda_{42}
&=\frac{t_{22}+\sqrt{t_{22}^2-4\abs{c_1c_2}^2\sin^42\a}}{2},    
\end{aligned}
\end{equation}
where
\begin{eqnarray}
\label{eq:t12}
t_{12}&:=&\abs{c_0}^2+\abs{c_3}^2-\cos^22\a(c_0c_3^*+c_3c_0^*),
\\\label{eq:t22}
t_{22}&:=&\abs{c_1}^2+\abs{c_2}^2+\cos^22\a(c_1c_2^*+c_2c_1^*).	
\end{eqnarray}

When $\b=\a$, one can show $\l_{11}$ is monotone decreasing in terms of $t_{11}$ while $\l_{21}$ is monotone increasing in terms of $t_{11}$, and $\l_{31},\l_{41}$ are monotone increasing in terms of $t_{21}$. Using the derivation rule of composite function, one can show $\l_{11},\l_{31},\l_{41}$ are monotone increasing while $\l_{21}$ is monotone decreasing in terms of $\a$. When $\a=0$, we have $\l_{11}=0,\l_{21}=\abs{c_0+c_3}^2,\l_{31}=0,\l_{41}=\abs{c_1-c_2}^2$. When $\a=\frac{\pi}{4}$, we have $\l_{11}=\abs{c_3}^2,\l_{21}=\abs{c_0}^2,\l_{31}=\abs{c_2}^2,\l_{41}=\abs{c_1}^2$. According to the monotonicities of $\l_{i1}$'s, one can show $[\abs{c_0}^2,\abs{c_1}^2,\abs{c_2}^2,\abs{c_3}^2]^T\prec[\l_{21},\l_{41},\l_{11},\l_{31}]^T$. From Lemma \ref{le:maj}, it implies $H(\l_{11},\l_{21},\l_{31},\l_{41})\leq H(\abs{c_0}^2,\abs{c_1}^2,\abs{c_2}^2,\abs{c_3}^2)$ with equality if and only if $\a=\frac{\pi}{4}$. Hence, we can exclude the condition $\a=\b$. So Eq. \eqref{eq:necon1} holds.

This completes the proof.
\end{proof}

Comparing the maximum of the boundary with the maximum of the non-boundary, we formulate $K_E(U)$ as follows.
\begin{proposition}
\label{pp:lowerbound}
For normalized operations $U$ satisfying $c_2=c_3$, i.e., $y=z$ in \eqref{eq:c0-c3exp},\\
\begin{equation}
\label{eq:pp-1main}
K_E(U)=
\begin{cases}
\max\limits_{\a\in[0,\p/4]}
\bigg\{
E(\ph(\a,\frac{\pi}{2}-\a;\frac{\pi}{2},\frac{\pi}{2})),1
\bigg\}, & \text{$\cos(2x+2y)\leq 0$},\\
\max\limits_{\a\in[0,\p/4]}
E(\ph(\a,\frac{\pi}{2}-\a;\frac{\pi}{2},\frac{\pi}{2})), & \text{$\cos(2x+2y)> 0$}.
\end{cases}
\end{equation}
\end{proposition}

\begin {remark}
\label{re:speU-1}
Since $E(\ph(\a,\frac{\pi}{2}-\a;\frac{\pi}{2},\frac{\pi}{2}))$ is a function of single one variable $\a$, we can plot $E(\ph(\a,\frac{\pi}{2}-\a;\frac{\pi}{2},\frac{\pi}{2}))$ and clearly observe the maximum from the plot when the three parameters $x,y,z$ are fixed. We construct several numerical examples in Fig. \ref{fig:prop1} to show Eq. \eqref{eq:pp-1main} is effective to obtain $K_E(U)$. 
Also, it is obvious that for a general $U$, i.e., $c_2\neq c_3$, the rhs of Eq. \eqref{eq:pp-1main} is a lower bound of its entangling power. Since $K_{Sch}(U)=E(\ph(\frac{\pi}{4},\frac{\pi}{4};\frac{\pi}{2},\frac{\pi}{2}))$, it implies this lower bound is tighter than the known lower bound Schmidt Strength $K_{Sch}(U)$ which is defined in Eq. \eqref{eq:kschU} and introduced in detail there.
\end{remark}

\begin{figure}[!h]
\centering
\includegraphics[height=8cm, width=16cm, angle=0]{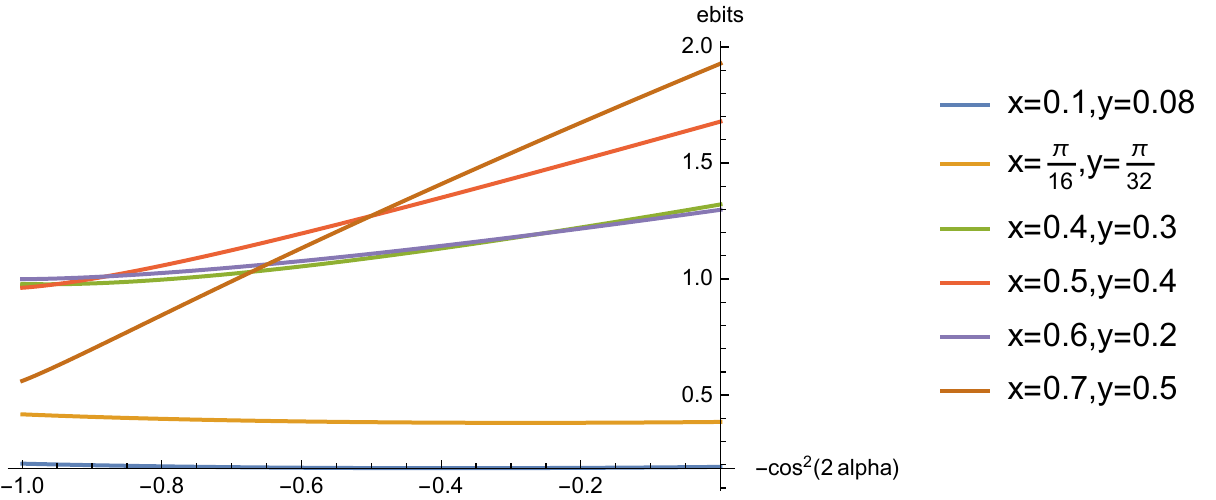}
\caption{Numerical examples of $K_E(U)$ with $y=z$.}
\label{fig:prop1}
\end{figure}

From Fig. \ref{fig:prop1} we find the maximums of $E(\ph(\a,\frac{\pi}{2}-\a;\frac{\pi}{2},\frac{\pi}{2}))$ occur only at the edge points, i.e., $\a$ equals $0$ or $\frac{\pi}{4}$. So we present the following conjecture.

\begin{conjecture}
\label{cj:prop1}
For normalized operations $U$ satisfying $c_2=c_3$, i.e., $y=z$ in \eqref{eq:c0-c3exp},\\
\begin{equation}
\label{eq:cjpp1}
\begin{aligned}
K_E(U)&=\max\{E(\ph(0,\frac{\pi}{2};\frac{\pi}{2},\frac{\pi}{2})),E(\ph(\frac{\pi}{4},\frac{\pi}{4};\frac{\pi}{2},\frac{\pi}{2}))\}\\
      &=\max\Big\{ H\big(\abs{c_0-c_3}^2,\abs{c_1+c_2}^2\big), H\big(\abs{c_0}^2,\abs{c_1}^2,\abs{c_2}^2,\abs{c_3}^2\big) \Big\}.
\end{aligned}
\end{equation}
\end{conjecture}
By computing one can show
\begin{eqnarray}
\label{eq:cj-1}
&H\big(\abs{c_0-c_3}^2,\abs{c_1+c_2}^2\big)=H\big(\cos^2(x+y),\sin^2(x+y)\big),\\
\label{eq:cj-2}
&H(\abs{c_0}^2,\abs{c_1}^2,\abs{c_2}^2,\abs{c_3}^2)=H(\cos^2 x\cos^4 y+\sin^2 x\sin^4 y,\cos^4 y \sin^2 x+\cos^2 x \sin^4 y,\\
\notag
&\cos^2 y \sin^2 y,\cos^2 y \sin^2 y).
\end{eqnarray}
For those $U$ whose parameters $x,y$ make Eq. \eqref{eq:cj-1} greater than Eq. \eqref{eq:cj-2}, the entangling power is strictly greater than Schmidt Strength.

In the following Sec. \ref{sec:specialU}, we will present two examples, for which we can analytically work out the entangling power $K_E(U)$. The analytical results support that Conjecture \ref{cj:prop1} holds. 

\section{Two examples of $U$ with $c_2=c_3$}
\label{sec:specialU}

\subsection{Example 1: $c_1=c_2=c_3$}
\label{subsubsec:f1}

In this subsection we consider the entangling power of $U$ as follows.
\begin{equation}
\label{eq:decomofspecU}
\begin{aligned}
U&=(\cos^3x+i\sin^3x)\s_0\ox\s_0+i\sin x\cos x e^{-ix}\sum_{j=1}^3 \s_j\ox\s_j\\
 &:=c_0\s_0\ox\s_0+c\sum_{j=1}^3 \s_j\ox\s_j,
\end{aligned}
\end{equation}
where $x\in(0,\frac{\pi}{4})$. One can verify that the expression of $U$ in \eqref{eq:decomofspecU} follows from setting $c_1=c_2=c_3$ in \eqref{eq:twoqubit}, namely $x=y=z$ in \eqref{eq:c0-c3exp}. Further, since $\abs{c_0}^2=1-3\abs{c}^2>\abs{c}^2$, we have $0<\abs{c}^2<\frac{1}{4}$. We will show the expression of entangling power of $U$ in Theorem \ref{thm:keuofspeu}. For this purpose, first of all, we present the following result.

\begin{proposition}
\label{pp:cristaspec1}
Suppose $U$ is expressed in Eq. \eqref{eq:decomofspecU}, 
\begin{equation}
\label{eq:profam1}
\begin{aligned}
\max\limits_{\a\in[0,\pi/4]\atop\b\in[0,\pi/2]}E(\ph(\a,\b;\frac{\pi}{2},\frac{\pi}{2}))&=\max\{E(\ph(\frac{\pi}{4},\frac{\pi}{4};\frac{\pi}{2},\frac{\pi}{2})),E(\ph(0,\frac{\pi}{2};\frac{\pi}{2},\frac{\pi}{2}))\},\\
&=\max\{H(4\abs{c}^2,1-4\abs{c}^2),H(1-3\abs{c}^2,\abs{c}^2,\abs{c}^2,\abs{c}^2)\}.
\end{aligned}
\end{equation}
\end{proposition}

The proof of Proposition \ref{pp:cristaspec1} is presented in Appendix \ref{sec:proofofpp2}.

\begin{theorem}
\label{thm:keuofspeu}
Suppose $U$ is expressed by \eqref{eq:decomofspecU}, and $x_0(\approx 0.1018)$ is the root of $H(\cos^22x_0,\sin^22x_0)=H(\cos^6x_0+\sin^6x_0,\cos^2x_0\sin^2x_0,\cos^2x_0\sin^2x_0,\cos^2x_0\sin^2x_0)$.
\\
(i) If $0 < x \leq x_0$, then $K_E(U)=H(\cos^22x,\sin^22x)$ ebits with the critical state $\ket{00}\ket{11}$.
\\
(ii) If $x_0\leq x\leq \frac{\pi}{4}$, then $K_E(U)=H(\cos^6x+\sin^6x,\cos^2x\sin^2x,\cos^2x\sin^2x,\cos^2x\sin^2x)$ ebits with the critical state $\frac{1}{\sqrt{2}}(\ket{00}+\ket{11})\otimes\frac{1}{\sqrt{2}}(\ket{00}+\ket{11})$.
\end{theorem}

\begin{proof}
For such $U$, we have $K_E(U)=\max\limits_{\a\in[0,\pi/4],\b\in[0,\pi/2]}E(\ph(\a,\b;\frac{\pi}{2},\frac{\pi}{2}))$ from Theorem \ref{pp:fam2-1}. Further
Proposition \ref{pp:cristaspec1} shows that $K_E(U)$ can be realized by either the local product state ($\a,\b=0,\frac{\pi}{2}$) or the local maximally entangled state ($\a=\b=\frac{\pi}{4}$). That is $K_E(U)=\max\{E(\ph(\frac{\pi}{4},\frac{\pi}{4};\frac{\pi}{2},\frac{\pi}{2})),E(\ph(0,\frac{\pi}{2};\frac{\pi}{2},\frac{\pi}{2}))\}$. Their parametric forms are 
\begin{equation}
\label{eq:lme+lp}
\begin{aligned}
E(\ph(\frac{\pi}{4},\frac{\pi}{4};\frac{\pi}{2},\frac{\pi}{2}))&=H(\cos^6x+\sin^6x,\cos^2x\sin^2x,\cos^2x\sin^2x,\cos^2x\sin^2x),\\
E(\ph(0,\frac{\pi}{2};\frac{\pi}{2},\frac{\pi}{2}))&=H(\cos^22x,\sin^22x).
\end{aligned}
\end{equation}
By computing we solve the equation $E(\ph(\frac{\pi}{4},\frac{\pi}{4};\frac{\pi}{2},\frac{\pi}{2}))=E(\ph(0,\frac{\pi}{2};\frac{\pi}{2},\frac{\pi}{2}))$ and obtain the root $x_0\approx 0.1018$. Then we have the local product state is the critical state when $x\in(0,x_0]$, and the local maximally entangled state is the critical state when $x\in(x_0,\frac{\pi}{4}]$. 

This completes the proof.
\end{proof}

Observe Theorem \ref{thm:keuofspeu} (i). When $0 < x < x_0$, we have 
\begin{equation}
\label{eq:ksch<ke}
\begin{aligned}
&K_{Sch}(U)=H(\cos^6x_0+\sin^6x_0,\cos^2x_0\sin^2x_0,\cos^2x_0\sin^2x_0,\cos^2x_0\sin^2x_0)\\
&<H(\cos^22x_0,\sin^22x_0)=K_E(U).
\end{aligned}
\end{equation}
Eq. \eqref{eq:ksch<ke} shows an explicit example of unitary $U$ whose Schmidt Strength $K_{Sch}(U)$ is strictly less than its entangling power $K_E(U)$.

There are infinitely many critical states in Theorem \ref{thm:keuofspeu}, because $(V_{R_A}\ox W_{R_B})\ket{\ps}$ is still a critical state when so is $\ket{\ps}$.  Comparing Theorem \ref{thm:keuofspeu} with the results in Ref. \cite{kc2000}, we conclude the critical states could be different if we select different measures for the same $U$ due to the two different critical points, i.e., $x_0$ here and $\a_0$ there \cite[Fig. 1]{kc2000}.

Finally one can show $U$ with $\abs{c_1}=\abs{c_2}=\abs{c_3}$ is equivalent to $U$ with $c_1=c_2=c_3$ by using Eqs. \eqref{eq:c02-c12}-\eqref{eq:c02-c32} and the constraint \eqref{eq:xyz}. So we have a slightly larger set of two-qubit unitary operations with computable entangling power as follows.

\begin{corollary}
\label{cr:c1=c2=c3}
The two-qubit unitary operation $U$ with $\abs{c_1}=\abs{c_2}=\abs{c_3}$ in \eqref{eq:twoqubit}
has the same entangling power as that in \eqref{eq:decomofspecU}, and thus can be computed via
Theorem \ref{thm:keuofspeu}.	
\end{corollary}

\subsection{Example 2: $c_0=ic_1^*$ and $c_2=c_3$}
\label{subsubsec:f2}
In this subsection we consider the entangling power of $U$ as follows.
\begin{equation}
\label{eq:decomofspecU-2}
\begin{aligned}
U&:=c_0\s_0\ox\s_0+c_j\sum_{j=1}^3 \s_j\ox\s_j,
\end{aligned}
\end{equation}
where 
\begin{equation}
\label{eq:c0-c3-f2}
\begin{aligned}
c_0&={1\over\sqrt2}(\cos^2 y+i \sin^2 y),
\\
c_1&={1\over\sqrt2}(\sin^2 y+i \cos^2 y),
\\
c_2&={1\over\sqrt2}(\sin y\cos y+i \sin y\cos y),
\\
c_3&={1\over\sqrt2}(\sin y\cos y+i \sin y\cos y),	
\end{aligned}
\end{equation}
for $y\in[0,\p/4)$.
One can verify that the coefficients in \eqref{eq:c0-c3-f2} follow from setting $x=\frac{\pi}{4},y=z$ in \eqref{eq:c0-c3exp}. We will show the expression of entangling power of such $U$ in Theorem \ref{thm:keuoffam-2} whose proof based on the following proposition.

\begin{proposition}
\label{pp:speUf3-1}
For normalized operations $U$ with $x=\frac{\pi}{4}$ in \eqref{eq:c0-c3exp}, $\max\limits_{\a\in[0,\pi/4] \atop \b\in[0,\pi/2]}E(\ph(\a,\b;\frac{\pi}{2},\frac{\pi}{2}))=E(\ph(\frac{\pi}{4},\frac{\pi}{4};\frac{\pi}{2},\frac{\pi}{2}))=H(\abs{c_0}^2,\abs{c_0}^2,\abs{c_2}^2,\abs{c_2}^2)$.
\end{proposition}

The proof of Proposition \ref{pp:speUf3-1} is presented in Appendix \ref{sec:proofofpp3}.

\begin{theorem}
\label{thm:keuoffam-2}
Suppose $U$ is expressed by Eqs. \eqref{eq:decomofspecU-2} and \eqref{eq:c0-c3-f2}.
\begin{equation}
\label{eq:keufam-2}
\begin{aligned} 
K_E(U)&=H(\abs{c_0}^2,\abs{c_0}^2,\abs{c_2}^2,\abs{c_2}^2) \quad ebits,\\
      &=H\big(\frac{1}{2}(\cos^4 y+\sin^4 y),\frac{1}{2}(\cos^4 y+\sin^4 y),\sin^2 y\cos^2 y,\sin^2 y\cos^2 y\big) \quad ebits,
\end{aligned}
\end{equation}
where $y\in(0,\frac{\pi}{4})$, with the critical state $\frac{1}{\sqrt{2}}(\ket{00}+\ket{11})\otimes\frac{1}{\sqrt{2}}(\ket{00}+\ket{11})$.
\end{theorem}

\begin{proof}
For such normalized operations $U$, Theorem \ref{pp:fam2-1} shows the critical states should be $\ket{\ph(\a,\b;\frac{\pi}{2},\frac{\pi}{2})}$, and $K_E(U)=\max\limits_{\a\in[0,\pi/4],\b\in[0,\pi/2]}E(\ph(\a,\b;\frac{\pi}{2},\frac{\pi}{2}))$. Further Proposition \ref{pp:speUf3-1} shows that $K_E(U)=H(\abs{c_0}^2,\abs{c_0}^2,\abs{c_2}^2,\abs{c_2}^2)$ which can be reached by setting the local maximally entangled state $(\a=\b=\frac{\pi}{4})$ as the input state. Then Eq. \eqref{eq:keufam-2} holds. 

This completes the proof.
\end{proof}

In the same way, one can show $U$ with $\abs{c_0}=\abs{c_1}$ and $c_2=c_3$ is equivalent to $U$ with $c_0=ic_1^*$ and $c_2=c_3$ by using Eqs. \eqref{eq:xyz} and \eqref{eq:c02-c12}. So we have another slightly larger set of two-qubit unitary operations with computable entangling power as follows.

\begin{corollary}
\label{cr:c0=c1&c2=c3}
The two-qubit unitary operation $U$ with $\abs{c_0}=\abs{c_1}$ and $c_2=c_3$ in \eqref{eq:twoqubit}
has the same entangling power as that expressed by Eqs. \eqref{eq:decomofspecU-2} and \eqref{eq:c0-c3-f2}, and thus can be computed via
Theorem \ref{thm:keuoffam-2}.	
\end{corollary}

\section{Entangling power of Schmidt-rank-two bipartite unitary operations}
\label{sec:Schr2}

So far we have investigated the entangling power of two-qubit unitary operations. It is evidently a harder problem to investigate the unitary operations in higher dimensions. In this section we investigate the entangling power of the Schmidt-rank-two bipartite unitary operation in arbitrary dimensions. Such unitary can be written as $V=\ketbra{1}{1}\otimes I_n+\ketbra{2}{2}\otimes\sum_{j=1}^n e^{i\theta_j}\ketbra{j}{j}$ where $\t_j$'s are real. It follows from \cite [Eq.(18)]{lcly20160808} that
\begin{equation}
\label{ep:schr2u1}
\mathit{K}_{\mathit{E}}(V)=\max_{c_1,\cdots,c_n\ge0,\atop \sum_j c_j=1} \mathit{H}
\bigg(
\frac{1-\big(1-4y(\{c_j\})\big)^{1\over2}}{2},\frac{1+\big(1-4y(\{c_j\})\big)^{1\over2}}{2}
\bigg),
\end{equation}
where
\begin{eqnarray}
\label{eq:ycj}	
y(\{c_j\})&\coloneqq &\sum_{j>k}c_jc_k\sin^2\big(\frac{\theta_j-\theta_k}{2}\big).
\end{eqnarray}
From the property of function $H(p,1-p)$, we know computing $K_E(V)$ is equivalent to maximizing $y(\{c_j\})$ over the conditions $c_1,\cdots,c_n\ge0$ and $\sum_j c_j=1$. The case $n=2$ has been investigated in Lemma 8 of \cite{lcly20160808}. In the following we study the case $n\geq3$. For this purpose, we introduce the conditional extremum characterized by the system of linear equations
\begin{equation}
\label{eq:extremum}
{\partial \big( y(\{c_j\}) + \l (\sum_j c_j-1)\big) \over \partial c_j}=0,
\end{equation}
where $\lambda$ is the Lagrange multiplier. For convenience, we formulate \eqref{eq:ycj}	 and  \eqref{eq:extremum} in matrix forms as follows.
\begin{eqnarray}
\label{eq:ycj-2}
y(\{c_j\})
=
{1\over2}
\bma
c_1 & c_2 & \cdots & c_n
\ema
\cdot
M_n
\cdot
\bma
c_1\\
c_2\\
\vdots\\
c_n
\ema,\\
\label{eq:Lagrangemultiplier2}
M_n
\cdot
\bma
c_1\\
c_2\\
\vdots\\
c_n
\ema
=-\lambda
\bma
1\\
1\\
\vdots\\
1
\ema,
\end{eqnarray}
where 
\begin{equation}
\label{eq:fam2}
M_n:=\bma
\sin^2\frac{\theta_1-\theta_1}{2} & \sin^2\frac{\theta_1-\theta_2}{2} & \cdots & \sin^2\frac{\theta_1-\theta_n}{2} \\
\sin^2\frac{\theta_2-\theta_1}{2} & \sin^2\frac{\theta_2-\theta_2}{2} & \cdots & \sin^2\frac{\theta_2-\theta_n}{2} \\
\vdots & \vdots & \vdots & \vdots \\
\sin^2\frac{\theta_n-\theta_1}{2} & \sin^2\frac{\theta_n-\theta_2}{2} & \cdots & \sin^2\frac{\theta_n-\theta_n}{2}
\ema.
\end{equation}
Since $M_n$ belongs to the family of matrices defined in Definition \ref{def:mat1}, we have $\rank M_n\leq 3,\forall n\geq 3$ from Lemma \ref{le:rleq3} in Appendix \ref{sec:rankofAn}.
To maximize $y(\{c_j\})$, we need to compare $y(\{c_j\})$ in the points fixed by \eqref{eq:Lagrangemultiplier2} with boundary points. We begin with the case $n=3$.
\begin{lemma}
\label{le:n=3}
For $n=3$, $\max y(\{c_j\})$ has the following two cases.

(i) If $\sin^2\frac{\theta_1-\theta_2}{2}\sin^2\frac{\theta_2-\theta_3}{2}\sin^2\frac{\theta_3-\theta_1}{2}>0$ and the vector $
\bma
\cos(\frac{\theta_2-\theta_3}{2})\csc(\frac{\theta_1-\theta_2}{2})\csc(\frac{\theta_1-\theta_3}{2})\\
\cos(\frac{\theta_1-\theta_3}{2})\csc(\frac{\theta_2-\theta_1}{2})\csc(\frac{\theta_2-\theta_3}{2})\\
\cos(\frac{\theta_1-\theta_2}{2})\csc(\frac{\theta_3-\theta_1}{2})\csc(\frac{\theta_3-\theta_2}{2})
\ema
$ has non-negative components, then $\max y(\{c_j\})=\frac{1}{4}$. 

(ii) Otherwise, $\max y(\{c_j\})=\frac{1}{4}\max\{\sin^2\frac{\theta_1-\theta_2}{2},\sin^2\frac{\theta_1-\theta_3}{2},\sin^2\frac{\theta_2-\theta_3}{2}\}$. 
\end{lemma}
\begin{proof}
(i) One can show $\det(M_3)=2\sin^2\frac{\theta_1-\theta_2}{2}\sin^2\frac{\theta_2-\theta_3}{2}\sin^2\frac{\theta_3-\theta_1}{2}>0$ which follows from hypothesis in (i). It implies $\rank M_3=3$. So for $n=3$, the system of linear equations \eqref{eq:Lagrangemultiplier2} has exact one solution $
\bma
\lambda\\
c_1\\
c_2\\
c_3
\ema=
\bma
-\frac{1}{2}\\
\frac{1}{2}\cos(\frac{\theta_2-\theta_3}{2})\csc(\frac{\theta_1-\theta_2}{2})\csc(\frac{\theta_1-\theta_3}{2})\\
\frac{1}{2}\cos(\frac{\theta_1-\theta_3}{2})\csc(\frac{\theta_2-\theta_1}{2})\csc(\frac{\theta_2-\theta_3}{2})\\
\frac{1}{2}\cos(\frac{\theta_1-\theta_2}{2})\csc(\frac{\theta_3-\theta_1}{2})\csc(\frac{\theta_3-\theta_2}{2})
\ema.
$ One can verify $c_1+c_2+c_3=1$. The hypothesis in (i) also insures $c_1,c_2,c_3$ satisfy the conditions $c_1,c_2,c_3\ge0$. In this point fixed by the solution, we have $y(\{c_j\})=-\frac{\lambda}{2}=\frac{1}{4}$ from \eqref{eq:ycj-2}. Hence, the assertion (i) holds.

(ii) Suppose $\theta_1,\theta_2,\theta_3$ do not satisfy the hypothesis in (i). There are two cases when maximizing $y(\{c_j\})$.

Case 1. If there exists one negative component in the above solution vector, the maximum of $y(\{c_j\})$ occurs on the boundary. Hence, we have $\max y(\{c_j\})=\frac{1}{4}\max\{\sin^2\frac{\theta_1-\theta_2}{2},\sin^2\frac{\theta_1-\theta_3}{2},\sin^2\frac{\theta_2-\theta_3}{2}\}$ which is implied from Propositon 1 \cite{lcly20160808}.

Case 2. When $\sin^2\frac{\theta_1-\theta_2}{2}\sin^2\frac{\theta_2-\theta_3}{2}\sin^2\frac{\theta_3-\theta_1}{2}=0$, straightforward computation shows $\max y(\{c_j\})=\frac{1}{4}\max\{\sin^2\frac{\theta_1-\theta_2}{2},\sin^2\frac{\theta_1-\theta_3}{2},\sin^2\frac{\theta_2-\theta_3}{2}\}$.

To sum up, the assertion (ii) holds.

This completes the proof.
\end{proof}

Based on Lemma \ref{le:n=3} we further investigate the entangling power of $V$ for any $n>3$ in \eqref{ep:schr2u1}.

\begin{theorem}
\label{thm:Schr2}
Suppose $V=\ketbra{0}{0}\otimes I_n+\ketbra{1}{1}\otimes\sum\limits_{j=1}^{n}e^{i\theta_j}\ketbra{j}{j}$, $\forall n>3$, is a
Schmidt-rank-two bipartite unitary operation, where $\t_j$'s are real.

(i) If there exist $\theta_{i_1},\theta_{i_2},\theta_{i_3}$ satisfying $\sin^2\frac{\theta_{i_1}-\theta_{i_2}}{2}\sin^2\frac{\theta_{i_2}-\theta_{i_3}}{2}\sin^2\frac{\theta_{i_3}-\theta_{i_1}}{2}>0$ , and non-negative $c_j$'s, $j\neq i_1,i_2,i_3$ such that vector
\begin{equation}
\label{eq:nonnegcompon-2}
\bma
c_{i_1}\\
c_{i_2}\\
c_{i_3}
\ema
=
\bma
\bigg(\cos(\frac{\theta_{i_2}-\theta_{i_3}}{2})-\sum\limits_{j\neq i_1,i_2,i_3}\sin\frac{\theta_j-\theta_{i_2}}{2}\sin\frac{\theta_j-\theta_{i_3}}{2}c_j\bigg)\csc(\frac{\theta_{i_1}-\theta_{i_2}}{2})\csc(\frac{\theta_{i_1}-\theta_{i_3}}{2})\\
\bigg(\cos(\frac{\theta_{i_1}-\theta_{i_3}}{2})-\sum\limits_{j\neq i_1,i_2,i_3}\sin\frac{\theta_j-\theta_{i_1}}{2}\sin\frac{\theta_j-\theta_{i_3}}{2}c_j\bigg)\csc(\frac{\theta_{i_2}-\theta_{i_1}}{2})\csc(\frac{\theta_{i_2}-\theta_{i_3}}{2})\\
\bigg(\cos(\frac{\theta_{i_1}-\theta_{i_2}}{2})-\sum\limits_{j\neq i_1,i_2,i_3}\sin\frac{\theta_j-\theta_{i_1}}{2}\sin\frac{\theta_j-\theta_{i_2}}{2}c_j\bigg)\csc(\frac{\theta_{i_3}-\theta_{i_1}}{2})\csc(\frac{\theta_{i_3}-\theta_{i_3}}{2})
\ema
\end{equation}
 has non-negative components, then $\mathit{K}_{\mathit{E}}(V)=1$ ebit.

(ii) Otherwise $\mathit{K}_{\mathit{E}}(V)=\max\limits_{1\leq i<j\leq n}\{h(i,j)\}$ ebits, where $h(i,j)\coloneqq\mathit{H}(\frac{1-\abs{\cos\frac{\theta_i-\theta_j}{2}}}{2},\frac{1+\abs{\cos\frac{\theta_i-\theta_j}{2}}}{2})$. 
\end{theorem}
\begin{proof}
(i) We know $\rank M_n\leq3$ from Lemma \ref{le:rleq3} in Appendix \ref{sec:rankofAn}. We have $\rank M_n=3$ which follows from the hypothesis $\sin^2\frac{\theta_{i_1}-\theta_{i_2}}{2}\sin^2\frac{\theta_{i_2}-\theta_{i_3}}{2}\sin^2\frac{\theta_{i_3}-\theta_{i_1}}{2}>0$. It implies there should be $(n-3)$ free variables in the fundamental system of solutions of the system of linear equations \eqref{eq:Lagrangemultiplier2}. We select $c_j$'s, $j\neq i_1,i_2,i_3$ as the $(n-3)$ free variables, and the remaining $c_{i_1},c_{i_2},c_{i_3}$ as the dependent variables. One can verify vectors $[c_1,c_2,\cdots,c_n]^T$ where $c_{i_1},c_{i_2},c_{i_3}$ are expressed in \eqref{eq:nonnegcompon-2} are the solutions of \eqref{eq:Lagrangemultiplier2}, and $\sum\limits_{i=1}^n c_i=1$.
The hypothesis in (i) assures $c_i$'s satisfy the conditions $c_1,\cdots,c_n\ge0$. In points fixed by the solutions, we have $y(\{c_j\})=-\frac{\lambda}{2}=\frac{1}{4}$ from \eqref{eq:ycj-2}. So $\mathit{K}_{\mathit{E}}(V)=1$ ebit from \eqref{ep:schr2u1}. Hence, the assertion (i) holds.

(ii) If the hypothesis in (i) doesn't hold, there are two cases when maximizing $y(\{c_j\})$.

Case 1. If there exist no such non-negative $c_j$'s, $j\neq i_1,i_2,i_3$ such that the vector in \eqref{eq:nonnegcompon-2} has non-negative components, the maximum of $y(\{c_j\})$ occurs on the boundary. So this case can be reduced to Case 1 in the proof of Lemma \ref{le:n=3} (ii) by setting some $(n-3)$ $c_j$'s equal to zero.  Hence, we have $\max y(\{c_j\})=\frac{1}{4}\max\limits_{1\leq {j_1}<{j_2}\leq n}\{\sin^2\frac{\theta_{j_1}-\theta_{j_2}}{2}\}$. Then we have 
$\mathit{K}_{\mathit{E}}(V)=\max\limits_{1\leq i<j\leq n}\{h(i,j)\}$ from \eqref{ep:schr2u1}. So the assertion (ii) holds.

Case 2. If $\sin^2\frac{\theta_{i_1}-\theta_{i_2}}{2}\sin^2\frac{\theta_{i_2}-\theta_{i_3}}{2}\sin^2\frac{\theta_{i_3}-\theta_{i_1}}{2}=0$ for any three pairwisely different $i_1,i_2,i_3$, this case can be reduced to Case 2 in the proof of Lemma \ref{le:n=3} (ii). Hence, the assertion (ii) holds from \eqref{ep:schr2u1}.

To sum up, the assertion (ii) holds.

This completes the proof.
\end{proof}

Using Lemma \ref{le:n=3}, Theorem \ref{thm:Schr2} and the facts in \cite[Eq. (16)]{lcly20160808}, we can analytically derive the entangling power of any Schmidt-rank-two bipartite unitary operations. Lemma \ref{le:n=3} and Theorem \ref{thm:Schr2} also show that Proposition 1 and Conjecture 1 in Ref. \cite{lcly20160808} are incomplete.


\section{Open problems}
\label{sec:open}

The first open problem from this paper is how to compute the entangling power of normalized $U$ of Schmidt rank four when $c_2\ne c_3$. The primary computation has shown that the analytical expression of $K_{E}(U)$ satisfies some monotonicity when $c_j$'s are in some interval. The second open problem is how to obtain the normal decomposition of bipartite unitary operations in higher dimensions, as it may decrease the number of parameters involved in the computation. Third, it is unknown whether we can extend our results to the assisted entangling power of $U$, namely the input states can be entangled or separable.

\section*{Acknowledgments}
\label{sec:ack}

We want to show our deepest gratitude to the anonymous referees for their careful work and useful suggestions. We thank Siddhartha Das for pointing out the paper \cite{sd1803}. 
YS and LC were supported by the  NNSF of China (Grant No. 11501024), Beijing Natural Science Foundation (4173076), and the Fundamental Research Funds for the Central Universities (Grant Nos. KG12040501, ZG216S1810 and ZG226S18C1).

\appendix

\section{The proof of Lemma \ref{le:boundarycases}}
\label{sec:proofoflemma4}

\begin{proof}
(i) Using those formulas in Sec. \ref{subsec:mathproperty} one can show that $\forall 1\leq j,k \leq 4$, $m_{jk}(\a,\b)=m_{kj}^*(\a,\b)$, $m_{jk}({\pi\over2}-\a,{\pi\over2}-\b)=m_{(5-j)(5-k)}(\a,\b)$ and
$\sum_j m_{jj}(\a,\b)=1$. Substituting ${\pi\over2}-\a$ and ${\pi\over2}-\b$ for $\a$ and $\b$ respectively in \eqref{eq:Ephialphabetax}, we convert $\rho_{BR_B}$ into $(\s_1\ox\s_1)\rho_{BR_B}(\s_1\ox\s_1)$. Then we have $E(\ph(\a,\b;\frac{\pi}{2},\frac{\pi}{2}))=E(\ph(\frac{\pi}{2}-\a,\frac{\pi}{2}-\b;\frac{\pi}{2},\frac{\pi}{2}))$. Hence, the assertion (i) holds.

(ii) We first claim $\max\limits_{\a\in[0,\frac{\pi}{4}]}E(\ph(\a,0;\frac{\pi}{2},\frac{\pi}{2}))=E(\ph(\frac{\pi}{4},0;\frac{\pi}{2},\frac{\pi}{2}))<\max\limits_{\a\in[0,\frac{\pi}{4}]}E(\ph(\a,\frac{\pi}{2};\frac{\pi}{2},\frac{\pi}{2}))$. When $\b=0$ or $\frac{\pi}{2}$, $E(\ph(\a,\b;\frac{\pi}{2},\frac{\pi}{2}))=H(t_1,t_2)$ where $t_1,t_2$ are expressed by \eqref{eq:t1t2}. To maximize $H(t_1,t_2)$ is to minimize $\abs{t_1-t_2}$. Using Eqs. \eqref{eq:cos2x} and \eqref{eq:c0c3c1c2} one can show
\begin{equation}
\label{eq:abst1-t2}
\begin{aligned}
\abs{t_1-t_2}&=\abs{2(c_0c_3^*+c_0^*c_3)\cos2\a\cos2\b+(\abs{c_0}^2+\abs{c_3}^2)-(\abs{c_1}^2+\abs{c_2}^2)},\\
&=\abs{(\sin2x\sin2y)\cos2\a\cos2\b+\cos2x\cos2y}.
\end{aligned}
\end{equation}
Since we assume $\a\in[0,\frac{\pi}{4}],\b\in[0,\frac{\pi}{2}]$, one can show the minimum of $\abs{t_1-t_2}$ when $\b=0$ is greater than when $\b=\frac{\pi}{2}$. So the first claim holds.

Secondly we claim that $\max\limits_{\b\in[0,\frac{\pi}{2}]}E(\ph(0,\b;\frac{\pi}{2},\frac{\pi}{2}))=\max\limits_{\a\in[0,\frac{\pi}{4}]}E(\ph(\a,\frac{\pi}{2};\frac{\pi}{2},\frac{\pi}{2}))=1$ when $\cos(2x+2y)\leq 0$, and $\max\limits_{\b\in[0,\frac{\pi}{2}]}E(\ph(0,\b;\frac{\pi}{2},\frac{\pi}{2}))=\max\limits_{\a\in[0,\frac{\pi}{4}]}E(\ph(\a,\frac{\pi}{2};\frac{\pi}{2},\frac{\pi}{2}))=E(\ph(0,\frac{\pi}{2};\frac{\pi}{2},\frac{\pi}{2}))$ when $\cos(2x+2y)> 0$. 
From Eq. \eqref{eq:abst1-t2}, one can verify $\max\limits_{\b\in[0,\frac{\pi}{2}]}E(\ph(0,\b;\frac{\pi}{2},\frac{\pi}{2}))=\max\limits_{\a\in[0,\frac{\pi}{4}]}E(\ph(\a,\frac{\pi}{2};\frac{\pi}{2},\frac{\pi}{2}))$. So it suffices to consider $\max\limits_{\b\in[0,\frac{\pi}{2}]}E(\ph(0,\b;\frac{\pi}{2},\frac{\pi}{2}))$. When $\a=0$, \eqref{eq:abst1-t2} can be simplified into
\begin{equation}
\label{eq:abst1-t2-2}
\begin{aligned}
\abs{t_1-t_2}&=\abs{2(c_0c_3^*+c_0^*c_3)\cos2\b+(\abs{c_0}^2+\abs{c_3}^2)-(\abs{c_1}^2+\abs{c_2}^2)},\\
&=\abs{(\sin2x\sin2y)\cos2\b+\cos2x\cos2y}.
\end{aligned}
\end{equation}
Since there is a restriction \eqref{eq:xyz} for parameters $x,y,z$, we find $(\sin2x\sin2y)\cos2\b+\cos2x\cos2y$ is a monotone increasing function with the variable $\cos2\b$. So we have $\min\limits_{\b\in[0,{\pi\over2}]}(\sin2x\sin2y)\cos2\b+\cos2x\cos2y=\cos(2x+2y)$. If $\cos(2x+2y)\leq 0$, it implies $\min\abs{t_1-t_2}=0$. Then we have $\max H(t_1,t_2)=1$ ebit. If $\cos(2x+2y)> 0$, it implies $\min\abs{t_1-t_2}=\cos(2x+2y)$ which occurs when $\b=\frac{\pi}{2}$. Then we have $\max H(t_1,t_2)=E(\ph(0,\frac{\pi}{2};\frac{\pi}{2},\frac{\pi}{2}))=H(\abs{c_0-c_3}^2,\abs{c_1-c_2}^2)$. So the second claim holds.

Finally we claim $\max\limits_{\b\in[0,\frac{\pi}{2}]}E(\ph(\frac{\pi}{4},\b;\frac{\pi}{2},\frac{\pi}{2}))=\max\limits_{\a\in[0,\frac{\pi}{4}]}E(\ph(\a,\frac{\pi}{4};\frac{\pi}{2},\frac{\pi}{2}))=E(\ph(\frac{\pi}{4},\frac{\pi}{4};\frac{\pi}{2},\frac{\pi}{2}))$. It follows from straightforward computation. 

Summarizing the above three claims, the assertion (ii) holds. 

This completes the proof.
\end{proof}

\section{The proof of Proposition \ref{pp:cristaspec1}}
\label{sec:proofofpp2}

\begin{proof}
From Proposition \ref{pp:lowerbound}, one can show \\
(i)If $x\in[\frac{\pi}{8},\frac{\pi}{4})$,
\begin{equation}
\label{eq:fam1pp1-1}
\max_{\a\in[0,\pi/4],\b\in[0,\pi/2]}E(\ph(\a,\b;\frac{\pi}{2},\frac{\pi}{2}))=
\max_{\a\in[0,\p/4]}
\big\{
E(\ph(\a,\frac{\pi}{2}-\a;\frac{\pi}{2},\frac{\pi}{2})),1
\big\}.
\end{equation}
(ii)If $x\in(0,\frac{\pi}{8})$,
\begin{equation}
\label{eq:fam1pp1-2}
\max_{\a\in[0,\pi/4],\b\in[0,\pi/2]}E(\ph(\a,\b;\frac{\pi}{2},\frac{\pi}{2}))=
\max_{\a\in[0,\p/4]}
E(\ph(\a,\frac{\pi}{2}-\a;\frac{\pi}{2},\frac{\pi}{2})).
\end{equation}
We need to show 
\begin{equation}
\label{eq:fam1pp1-3}
\max_{\a\in[0,\p/4]}
E(\ph(\a,\frac{\pi}{2}-\a;\frac{\pi}{2},\frac{\pi}{2}))
=\max\{E(\ph(\frac{\pi}{4},\frac{\pi}{4};\frac{\pi}{2},\frac{\pi}{2})),E(\ph(0,\frac{\pi}{2};\frac{\pi}{2},\frac{\pi}{2}))\},
\end{equation}
and when $x\in[\frac{\pi}{8},\frac{\pi}{4})$,
\begin{equation}
\label{eq:fam1pp1-4}
E(\ph(\frac{\pi}{4},\frac{\pi}{4};\frac{\pi}{2},\frac{\pi}{2}))\geq 1.
\end{equation}

We first prove \eqref{eq:fam1pp1-3} holds. Using $c_1=c_2=c_3$, one can further simplify $\l_{i2}$'s which have been expressed in \eqref{eq:Ephialphabetaeigen2} as follows.
\begin{equation}
\label{eq:fam1case2lam}
\begin{aligned}
\l_{12}&=\frac{\abs{c_0}^2-\cos4\a\abs{c}^2-\sqrt{(\abs{c_0}^2-\abs{c}^2)(\abs{c_0}^2-\cos^24\a\abs{c}^2)}}{2},\\
\l_{22}&=\frac{\abs{c_0}^2-\cos4\a\abs{c}^2+\sqrt{(\abs{c_0}^2-\abs{c}^2)(\abs{c_0}^2-\cos^24\a\abs{c}^2)}}{2},\\
\l_{32}&=(\cos2\a-1)^2\abs{c}^2,\\
\l_{42}&=(\cos2\a+1)^2\abs{c}^2.\\
\end{aligned}
\end{equation}
Set $y=-\cos4\a$ and $E_2(y):=E(\ph(\a,\frac{\pi}{2}-\a;\frac{\pi}{2},\frac{\pi}{2}))$. One can show the derivative of $E_2(y)$ as follows.
\begin{equation}
\label{eq:partialderivatives41}
\begin{aligned}
{d E_2(y)
\over
d y}&=\frac{\abs{c}^2}{2}\bigg(-\sqrt{\frac{2}{1-y}}(\log\l_{32}-\log\l_{42})-\sqrt{\frac{1-4\abs{c}^2}{1-\abs{c}^2(3+y^2)}}y(\log\l_{12}-\log\l_{22})\\
&+\log\frac{\abs{c}^2}{1-3\abs{c}^2}\bigg).
\end{aligned}
\end{equation}
From \eqref{eq:partialderivatives41}, one can show $\lim\limits_{y\to-1}{d E_2(y)
\over
d y}=\abs{c}^2(\log4\abs{c}^2-\log(1-4\abs{c}^2))$, and $\lim\limits_{y\to1}{d E_2(y)
\over
d y}=\frac{2\abs{c}^2}{\ln2}\geq0$. We have $\lim\limits_{y\to-1}{d E_2(y)
\over
d y}<0$ when $\abs{c}^2\in(0,\frac{1}{8})$, and $\lim\limits_{y\to-1}{d E_2(y)
\over
d y}\geq 0$ when $\abs{c}^2\in[\frac{1}{8},\frac{1}{4})$. We first consider the case when $\abs{c}^2\in[\frac{1}{8},\frac{1}{4})$. Suppose $F_1(y,\abs{c}^2):={d E_2(y)
\over
d y}$. Fig. \ref{fig:fam1case1} shows $F_1(y,\abs{c}^2)$ is nonnegative within the domain $y\in[-1,1],\abs{c}^2\in[\frac{1}{8},\frac{1}{4})$. It implies when $\abs{c}^2\in[\frac{1}{8},\frac{1}{4})$, $E_2(y)$ is monotone increasing in terms of $y$. Hence, we conclude the maximum of $E_2(y)$ occurs when $y=1$, i.e., $\a=\frac{\pi}{4}$ if $\abs{c}^2\in[\frac{1}{8},\frac{1}{4})$.

We next consider the case when $\abs{c}^2\in(0,\frac{1}{8})$. 
Suppose $F_2(y,\abs{c}^2):={d^2 E_2(y)
\over
d y^2}$. Fig. \ref{fig:fam1case2} shows $F_2(y,\abs{c}^2)$ is nonnegative within the domain $y\in[-1,1],\abs{c}^2\in(0,\frac{1}{8})$. It implies $E_2(y)$ is convex when $\abs{c}^2\in(0,\frac{1}{8})$. Hence, we conclude the maximum of $E_2(y)$ occurs at the two edge points $y=-1,1$, i.e., $\a=0,\frac{\pi}{4}$ if $\abs{c}^2\in(0,\frac{1}{8})$. It implies $\max\limits_{\a\in[0,\frac{\pi}{4}]}E(\ph(\a,\frac{\pi}{2}-\a;\frac{\pi}{2},\frac{\pi}{2}))=\max\{E_2(-1),E_2(1)\}$ in this case. Combining the two cases, we have Eq.\eqref{eq:fam1pp1-3} holds.

Then we need to show when $x\in[\frac{\pi}{8},\frac{\pi}{4})$, Eq. \eqref{eq:fam1pp1-4} holds. By computing, we have $h(x):=E(\ph(\frac{\pi}{4},\frac{\pi}{4};\frac{\pi}{2},\frac{\pi}{2}))=H(\cos^6x+\sin^6x,\cos^2x\sin^2x,\cos^2x\sin^2x,\cos^2x\sin^2x)$. One can show $h(x)$ is monotone increasing. So we have $h(x)$ is lower bounded by $h(\frac{\pi}{8})\approx 1.55>1$. Therefore, Eq. \eqref{eq:fam1pp1-4} holds.

This completes the proof.
\end{proof}

\begin{figure}[!h]
\centering
\subfigure[$F_1(y,\abs{c}^2)$]
{
     \includegraphics[width=3.0in]{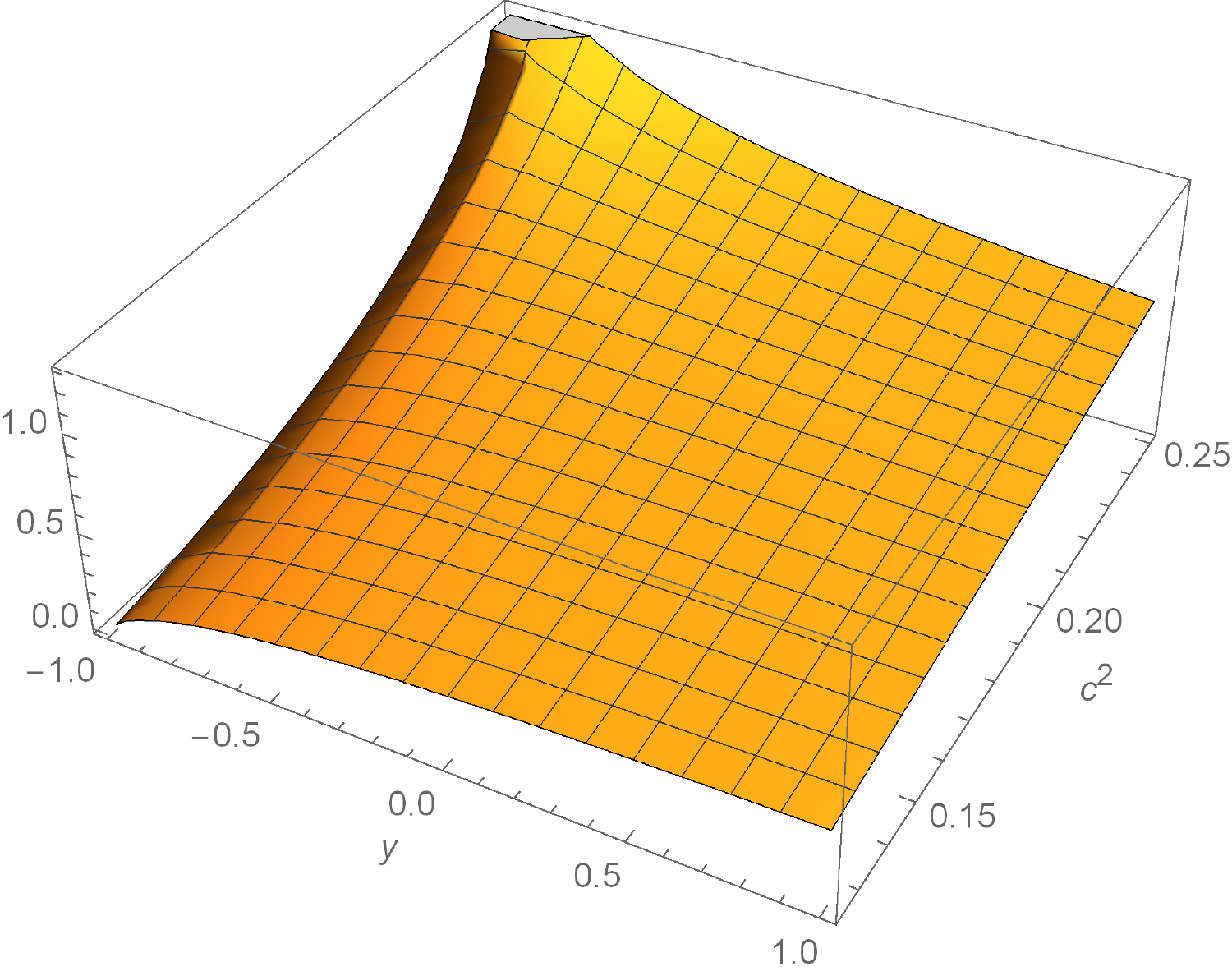}
     \label{fig:fam1case1}
}
\subfigure[$F_2(y,\abs{c}^2)$]
{
     \includegraphics[width=3.0in]{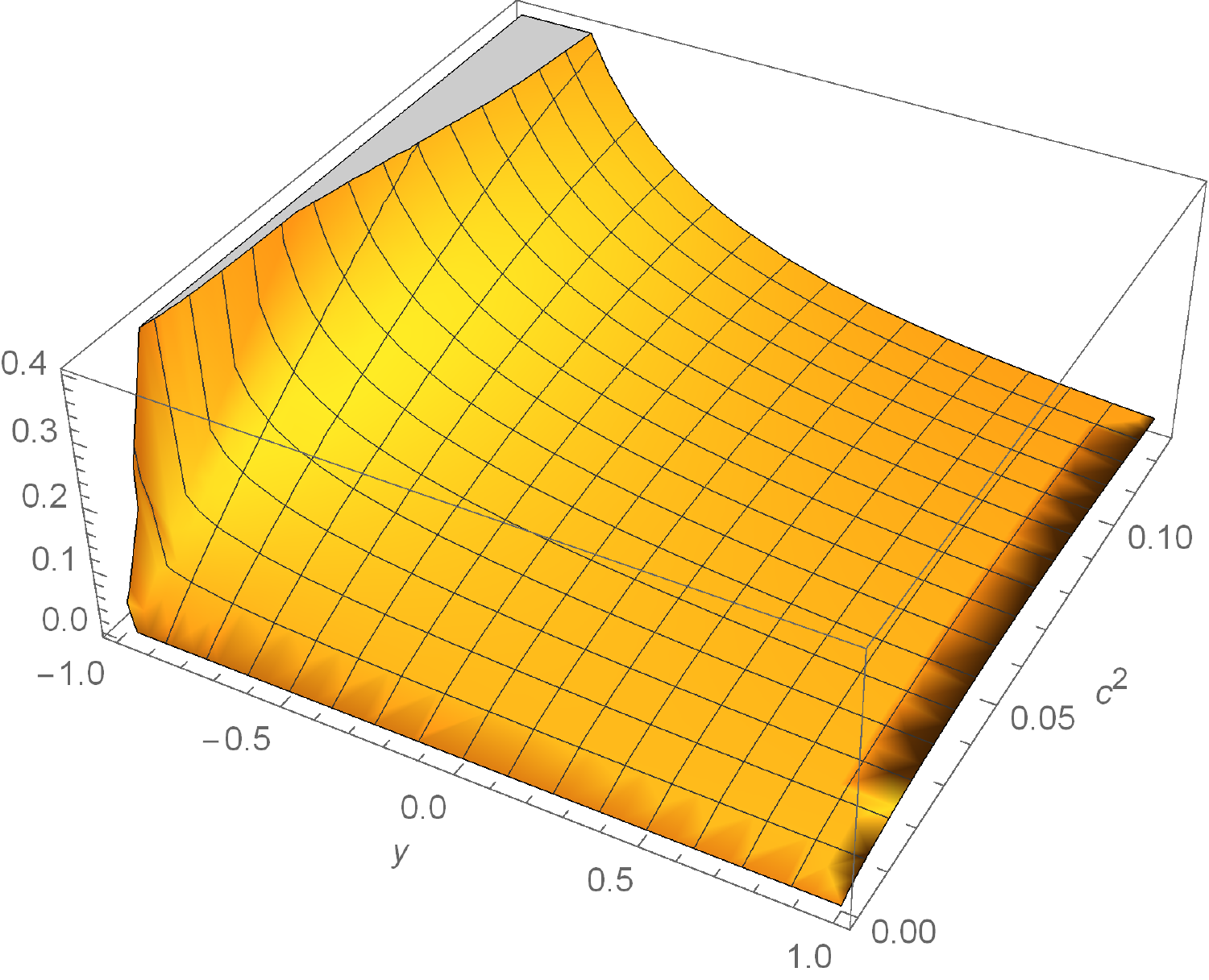}
     \label{fig:fam1case2}
}
\caption{Images of derivative}
\label{fig:family1}
\end{figure}

\section{The proof of Proposition \ref{pp:speUf3-1}}
\label{sec:proofofpp3}

\begin{proof}
When $x=\frac{\pi}{4}$, one can find the following fact.
\begin{equation}
\label{eq:c0=ic1*}
\begin{aligned}
&c_0=ic_1^*,\quad
 c_2=ic_3^*,\\
&\abs{c_0}=\abs{c_1}>
 \abs{c_2}=\abs{c_3},\\
&\abs{c_0}^2+\abs{c_2}^2
 ={1\over2}.
\end{aligned}
\end{equation}
One can show $H(\abs{c_0}^2,\abs{c_0}^2,\abs{c_2}^2,\abs{c_2}^2)\geq1$ by straightforward computation. From Proposition \ref{pp:lowerbound}, we only need to prove $\max\limits_{\a\in[0,\pi/4]}E(\ph(\a,\frac{\pi}{2}-\a;\frac{\pi}{2},\frac{\pi}{2}))=E(\ph(\frac{\pi}{4},\frac{\pi}{4};\frac{\pi}{2},\frac{\pi}{2}))$. Set $u=\cos^22\a\in[0,1]$ and $E(u):=E(\ph(\a,\frac{\pi}{2}-\a;\frac{\pi}{2},\frac{\pi}{2}))$. One can show that 
\begin{equation}
\label{eq:derioff3-1}
{d (\l_{22}+\l_{42})
\over
d u}
=\frac{1}{2}\bigg(\frac{-k+8l(1-u)+2k^2u}{\sqrt{(1-2ku)^2-16l(1-u)^2}}+\frac{k+8l(1-u)+2k^2u}{\sqrt{(1+2ku)^2-16l(1-u)^2}}\bigg),
\end{equation}
where
\begin{equation}
\label{eq:derioff3-2}
k=c_0c_3^*+c_0^*c_3=c_1c_2^*+c_1^*c_2=\frac{1}{2}\sin2y, \quad l=\abs{c_0c_3}^2=\abs{c_1c_2}^2.
\end{equation}
We claim 
$
{d (\l_{22}+\l_{42})
\over
d u}
\geq 0
$. One can show
\begin{equation}
\label{eq:derioff3-2}
\begin{aligned}
&\bigg(\frac{-k+8l(1-u)+2k^2u}{\sqrt{(1-2ku)^2-16l(1-u)^2}}/\frac{k+8l(1-u)+2k^2u}{\sqrt{(1+2ku)^2-16l(1-u)^2}}\bigg)^2\\
&=\frac{(k+8l(-1+u)-2k^2u)^2(-16l(-1+u)^2+(1+2ku)^2)}{(k-8l(-1+u)+2k^2u)^2(-16l(-1+u)^2+(1-2ku)^2)}.
\end{aligned}
\end{equation}
Considering the second term of \eqref{eq:derioff3-2}, the difference between the numerator and the denominator is $32kl(1-u)\big((4k^2-16l)u-1+16l\big)$. One can show this difference is non-positive from $4k^2-16l\leq 0,16l-1\leq0$. It implies $\abs{\frac{-k+8l(1-u)+2k^2u}{\sqrt{(1-2ku)^2-16l(1-u)^2}}/\frac{k+8l(1-u)+2k^2u}{\sqrt{(1+2ku)^2-16l(1-u)^2}}}\leq 1$. Since $\frac{k+8l(1-u)+2k^2u}{\sqrt{(1+2ku)^2-16l(1-u)^2}}$ is positive, we have \eqref{eq:derioff3-1} is nonnegative. It implies $(\l_{22}+\l_{42})$ is monotone increasing with $u\in[0,1]$. Since $(\l_{22}+\l_{42})=2\abs{c_0}^2$ when $u=0$, the monotonicity implies $(\l_{22}+\l_{42})\geq2\abs{c_0}^2$. Further we have $(\l_{12}+\l_{32})\leq2\abs{c_2}^2$. Therefore, one can show one of $\l_{22},\l_{42}$ isn't less than $\abs{c_0}^2$, and one of $\l_{12},\l_{32}$ isn't greater than $\abs{c_2}^2$. It implies $[\abs{c_0}^2,\abs{c_0}^2,\abs{c_2}^2,\abs{c_2}^2]^T\prec[\l_{12},\l_{22},\l_{32},\l_{42}]^T$. From Lemma \ref{le:maj}, we have $H(\l_{12},\l_{22},\l_{32},\l_{42})$ is upper bounded by $H(\abs{c_0}^2,\abs{c_0}^2,\abs{c_2}^2,\abs{c_2}^2)$ which is attained when $\a=\frac{\pi}{4}$. 

This completes the proof.
\end{proof}

\section{The rank of matrix $A_n$}
\label{sec:rankofAn}

\begin{definition}
\label{def:mat1}
Define a family of order-$n$ matrices as follows. Let $\a_j,\b_j\in\bbR$ for any $j$ and
\begin{equation}
\label{eq:fam1}
A_n:=\bma
\sin^2\alpha_1 & \sin^2(\alpha_1+\beta_1) & \cdots & \sin^2(\alpha_1+\beta_{n-1})\\
\sin^2\alpha_2 & \sin^2(\alpha_2+\beta_1) & \cdots & \sin^2(\alpha_2+\beta_{n-1})\\
\vdots & \vdots & \vdots & \vdots \\
\sin^2\alpha_n & \sin^2(\alpha_n+\beta_1) & \cdots & \sin^2(\alpha_n+\beta_{n-1})
\ema.
\end{equation}
\qed
\end{definition}

\begin{lemma}
\label{le:rleq3}
For any $n\geq 3$,
$
\rank A_n \le 3
$.

\end{lemma}
\begin{proof}
One can verify that $\sin^2(\alpha_i+\beta_j)=\sin^2\alpha_i\cos^2\beta_j+\cos^2\alpha_i\sin^2\beta_j+\frac{1}{2}\sin2\alpha_i\sin2\beta_j$. Then one can show $A_n$ is the sum of the three $\rank$ $1$ matrices as follows.
\begin{equation}
\begin{aligned}
\label{eq:sumrank1}
A_n&=
\bma
\sin^2\alpha_1 & \sin^2\alpha_1\cos^2\beta_1 & \cdots & \sin^2\alpha_1\cos^2\beta_{n-1} \\
\sin^2\alpha_2 & \sin^2\alpha_2\cos^2\beta_1 & \cdots & \sin^2\alpha_2\cos^2\beta_{n-1} \\
\vdots & \vdots & \vdots \\
\sin^2\alpha_n & \sin^2\alpha_n\cos^2\beta_1 & \cdots & \sin^2\alpha_n\cos^2\beta_{n-1}
\ema\\
&+
\bma
0 & \cos^2\alpha_1\sin^2\beta_1 & \cdots & \cos^2\alpha_1\sin^2\beta_{n-1} \\
0 & \cos^2\alpha_2\sin^2\beta_1 & \cdots & \cos^2\alpha_1\sin^2\beta_{n-1} \\
\vdots & \vdots & \vdots \\
0 & \cos^2\alpha_n\sin^2\beta_1 & \cdots & \cos^2\alpha_1\sin^2\beta_{n-1}
\ema\\
&+
\bma
0 & \frac{1}{2}\sin2\alpha_1\sin2\beta_1 & \cdots & \frac{1}{2}\sin2\alpha_1\sin2\beta_{n-1} \\
0 & \frac{1}{2}\sin2\alpha_2\sin2\beta_1 & \cdots & \frac{1}{2}\sin2\alpha_2\sin2\beta_{n-1} \\
\vdots & \vdots & \vdots \\
0 & \frac{1}{2}\sin2\alpha_n\sin2\beta_1 & \cdots & \frac{1}{2}\sin2\alpha_n\sin2\beta_{n-1}
\ema.
\end{aligned}
\end{equation}

So we have $\rank(A_n)\leq 3$. This completes the proof.
\end{proof}

\renewcommand\refname{References}
\bibliographystyle{ieeetr} 
\bibliography{dis}

\begin{thebibliography}{10}

\bibitem{bbc93}
C.~H. Bennett, G.~Brassard, C.~Cr\'epeau, R.~Jozsa, A.~Peres, and W.~K.
  Wootters, ``Teleporting an unknown quantum state via dual classical and
  einstein-podolsky-rosen channels,'' {\em Phys. Rev. Lett.}, vol.~70,
  pp.~1895--1899, Mar 1993.

\bibitem{ea91}
A.~K. Ekert, ``Quantum cryptography based on bell's theorem,'' {\em
  Phys.Rev.Lett.}, vol.~67, pp.~661--663, Aug 1991.

\bibitem{zmc16}
H.~Zhu, M.~Hayashi, and L.~Chen, ``Universal steering criteria,'' {\em Phys.
  Rev. Lett.}, vol.~116, p.~070403, Feb 2016.

\bibitem{lsw09}
N.~Linden, J.~A. Smolin, and A.~Winter, ``Entangling and disentangling power of
  unitary transformations are not equal,'' {\em Phys. Rev. Lett.}, vol.~103,
  p.~030501, Jul 2009.

\bibitem{sm10}
A.~Soeda and M.~Murao, ``Delocalization power of global unitary operations on
  quantum information,'' {\em New Journal of Physics}, vol.~12, no.~9,
  p.~093013, 2010.

\bibitem{mkz13}
M.~Musz, M.~Kus, and K.~Zyczkowski, ``Unitary quantum gates, perfect entanglers
  and unistochastic maps,'' {\em Phys. Rev. A}, vol.~87, no.~2, p.~022111,
  2013.

\bibitem{cy16}
L.~Chen and L.~Yu, ``Entanglement cost and entangling power of bipartite
  unitary and permutation operators,'' {\em Phys. Rev. A}, vol.~93, p.~042331,
  Apr 2016.

\bibitem{lcly20160808}
L.~Chen and L.~Yu, ``Entangling and assisted entangling power of bipartite
  unitary operations,'' {\em Phys. Rev. A}, vol.~94, p.~022307, Aug 2016.

\bibitem{Nielsen03}
M.~A. Nielsen, C.~M. Dawson, J.~L. Dodd, A.~Gilchrist, D.~Mortimer, T.~J.
  Osborne, M.~J. Bremner, A.~W. Harrow, and A.~Hines, ``Quantum dynamics as a
  physical resource,'' {\em Phys. Rev. A}, vol.~67, p.~052301, May 2003.

\bibitem{cy15}
L.~Chen and L.~Yu, ``Decomposition of bipartite and multipartite unitary gates
  into the product of controlled unitary gates,'' {\em Phys. Rev. A}, vol.~91,
  p.~032308, Mar 2015.

\bibitem{wootters1998}
W.~K. Wootters, ``Entanglement of formation of an arbitrary state of two
  qubits,'' {\em Phys. Rev. Lett.}, vol.~80, p.~2245, 1998.

\bibitem{kc2000}
B.~Kraus and J.~I. Cirac, ``Optimal creation of entanglement using a two-qubit
  gate,'' {\em Phys. Rev. A}, vol.~63, p.~062309, May 2001.

\bibitem{zj03}
J.~Zhang, J.~Vala, S.~Sastry, and K.~B. Whaley, ``Geometric theory of nonlocal
  two-qubit operations,'' {\em Phys. Rev. A}, vol.~67, p.~042313, Apr 2003.

\bibitem{ra04}
A.~T. Rezakhani, ``Characterization of two-qubit perfect entanglers,'' {\em
  Phys. Rev. A}, vol.~70, p.~052313, Nov 2004.

\bibitem{zj04}
J.~Zhang, J.~Vala, S.~Sastry, and K.~B. Whaley, ``Minimum construction of
  two-qubit quantum operations,'' {\em Phys. Rev. Lett.}, vol.~93, p.~020502,
  Jul 2004.

\bibitem{vd04}
G.~Vidal and C.~M. Dawson, ``Universal quantum circuit for two-qubit
  transformations with three controlled-not gates,'' {\em Phys. Rev. A},
  vol.~69, p.~010301, Jan 2004.

\bibitem{svv04}
V.~V. Shende, S.~S. Bullock, and I.~L. Markov, ``Recognizing small-circuit
  structure in two-qubit operators,'' {\em Phys. Rev. A}, vol.~70, p.~012310,
  Jul 2004.

\bibitem{vfwc04}
F.~Vatan and C.~Williams, ``Optimal quantum circuits for general two-qubit
  gates,'' {\em Phys. Rev. A}, vol.~69, p.~032315, Mar 2004.

\bibitem{kb10}
B.~Kraus, ``Local unitary equivalence and entanglement of multipartite pure
  states,'' {\em Phys. Rev. A}, vol.~82, p.~032121, Sep 2010.

\bibitem{ydy13}
N.~Yu, R.~Duan, and M.~Ying, ``Five two-qubit gates are necessary for
  implementing the toffoli gate,'' {\em Phys. Rev. A}, vol.~88, p.~010304, Jul
  2013.

\bibitem{cy13}
S.~M. Cohen and L.~Yu, ``{All unitaries having operator Schmidt rank 2 are
  controlled unitaries},'' {\em Phys. Rev. A}, vol.~87, p.~022329, Feb 2013.

\bibitem{zj040416}
J.~Zhang, J.~Vala, S.~Sastry, and K.~B. Whaley, ``Optimal quantum circuit
  synthesis from controlled-unitary gates,'' {\em Phys. Rev. A}, vol.~69,
  p.~042309, Apr 2004.

\bibitem{sd1803}
S.~Das, S.~Baumal, and M.~M. Wilde, ``{Entanglement and secret-key-agreement
  capacities of bipartite quantum interactions and read-only memory devices},''
  {\em ArXiv e-prints}, Dec. 2017.

\bibitem{mc2011}
M.~Hayashi and L.~Chen, ``Weaker entanglement between two parties guarantees
  stronger entanglement with a third party,'' {\em Phys. Rev. A}, vol.~84,
  p.~012325, Jul 2011.

\bibitem{nc2000book}
M.~A. Nielsen and I.~L. Chuang, {\em Quantum Computation and Quantum
  Information}.
\newblock Cambridge: Cambridge University Press, 2000.

\bibitem{MAG10}
M.~Ali, A.~R.~P. Rau, and G.~Alber, ``Quantum discord for two-qubit $x$
  states,'' {\em Phys. Rev. A}, vol.~81, p.~042105, Apr 2010.

\bibitem{CQXstates11}
Q.~Chen, C.~Zhang, S.~Yu, X.~X. Yi, and C.~H. Oh, ``Quantum discord of
  two-qubit $x$ states,'' {\em Phys. Rev. A}, vol.~84, p.~042313, Oct 2011.

\bibitem{LV01}
L.~Henderson and V.~Vedral, ``Classical, quantum and total correlations,'' {\em
  J. Phys. A: Math. Gen.}, vol.~34, p.~6899, Aug 2001.

\bibitem{VVcorrelations03}
V.~Vedral, ``Classical correlations and entanglement in quantum measurements,''
  {\em Phys. Rev. Lett.}, vol.~90, p.~050401, Feb 2003.

\bibitem{Chen2017Separability}
L.~Chen, K.~H. Han, and S.~H. Kye, ``Separability criterion for three-qubit
  states with a four dimensional norm,'' {\em J. Phys. A: Math. Gen.}, vol.~50,
  2017.

\bibitem{Han2016Construction}
K.~H. Han and S.~H. Kye, ``Construction of multi-qubit optimal genuine
  entanglement witnesses,'' {\em J. Phys. A: Math. Gen.}, vol.~49, no.~17,
  p.~175303, 2016.

\bibitem{Mendonca2015Maximally}
P.~E. M.~F. Mendonca, S.~M.~H. Rafsanjani, D.~Galetti, and M.~A. Marchiolli,
  ``Maximally genuine multipartite entangled mixed x-states of n-qubits,'' {\em
  J. Phys. A: Math. Gen.}, vol.~48, no.~21, 2015.

\end{thebibliography}

\end{document}